\documentclass[12pt]{article}
\usepackage[latin1]{inputenc}
\usepackage[english]{babel}
\usepackage{amsfonts,amsmath,mathrsfs,amssymb}
\usepackage[dvips]{graphicx}
\usepackage[sort, round]{natbib}
\usepackage{verbatim}
\usepackage{booktabs}
\usepackage{subfig}
\usepackage[usenames,dvipsnames]{color}
\usepackage{caption}
\usepackage[english]{babel}
\usepackage[lined,boxed,commentsnumbered]{algorithm2e}
\usepackage{multirow}
\usepackage[top=2.5cm,left=2.5cm,right=2.5cm,bottom=2.8cm]{geometry}
\usepackage{hyperref}

\def\T{{\footnotesize {^{_{\sf T}}}}} 
\newcommand{\Real}{{\rm I}\negthinspace {\rm R}}

\usepackage{color}

\newtheorem{theorem}{Theorem}[section]

\newtheorem{proposition}[theorem]{Proposition}

\newenvironment{proof}[1][Proof]{\begin{trivlist}
\item[\hskip \labelsep {\bfseries #1}]}{\end{trivlist}}

\begin{document}
%%%%%%%%%%%%%%%%%%%%%%%%%%%%%%%%%%%%%%%%%%%%%%%%%%%%%%%%%%%%%%%%

\title{Approximate Bayesian Computation\\ with composite score functions}

\author{Erlis Ruli$^*$, Nicola Sartori and  Laura Ventura\\
{\it {\small Department of Statistical Sciences, University of Padova, Italy}} \\
{\tt {\small $^*$ruli@stat.unipd.it, sartori@stat.unipd.it, ventura@stat.unipd.it}} }
\maketitle
\begin{abstract}
Both Approximate Bayesian Computation (ABC)  and composite likelihood methods are useful for Bayesian and frequentist inference, respectively, when the likelihood function is intractable. 
We propose to use composite likelihood score functions as summary statistics in ABC in order to obtain accurate approximations to the posterior distribution. This is  motivated by the use of the score function of the full likelihood, and extended to general unbiased estimating functions in complex models.
Moreover, we show that  if the composite score is suitably standardised, the resulting ABC procedure is  invariant to reparameterisations and automatically adjusts the curvature of the composite likelihood, and of the corresponding posterior distribution.
The method is illustrated through examples with simulated data, and an  application to modelling of spatial extreme rainfall data is discussed.
\end{abstract}


%%%%%%%%%%%%%%%%%%%%%%%%%%%%%%%%%%%%%%%%%%%%%%%%%%%%%%%%%%%%%%%%
\section{Introduction}
%%%%%%%%%%%%%%%%%%%%%%%%%%%%%%%%%%%%%%%%%%%%%%%%%%%%%%%%%%%%%%%%

The summary of the data on a given model offered by the likelihood function is the key ingredient of all likelihood-based inferential methods. However, likelihood inference, both frequentist and Bayesian, { is difficult or even impossible to perform} when the likelihood is analytically or computationally intractable. This usually occurs in the presence of complex models, such as models with complicated dependence structures or in models with many latent variables. 
 
In these situations, for frequentist or Bayesian inference, surrogates of the ordinary likelihood are useful. A notable contribution is given by composite likelihoods {\citep{lindsay1988composite}}, which are based on the composition of suitable lower dimensional densities, such as bivariate marginal \citep{cox2004note}, conditional or full conditional densities \citep{varin2011overview}. The use of composite likelihoods has been widely advocated in different complex applications of frequentist inference (see \citealp{varin2011overview}, for a general review, and \citealp{larribe2011composite}, for a review in genetics).

%From a Bayesian perspective, when the computation of the likelihood is impracticable but it is easy to simulate from the model, Approximate Bayesian Computation (ABC) may be used to obtain an approximation for the posterior distribution. 
From a Bayesian perspective, when the computation of the likelihood is impracticable, but it is easy to simulate from the model, {an approximation of the posterior distribution can be obtained by Approximate Bayesian Computation (ABC)}.
The idea {of ABC} is to simulate from the model for different parameter values, and to keep those values that produce simulated datasets that approximately match the observed data \citep[see][]{beaumont2010,marin2012}. The most popular {ABC} approach  is to consider an approximate matching of some summary statistics, evaluated at the observed and simulated data, by means of suitable distances. When the statistics are sufficient for the parameters of the model, this method  leads to the exact posterior distribution as the distance tends to zero. However, in realistic  applications sufficient statistics are not available and the practitioner must resort to a careful selection of data summaries. 

In this paper we propose the use of a {scaled} composite likelihood score function as  summary statistic in ABC. The idea is  motivated by the use of the score function when the full likelihood is available and is then extended to  composite likelihood score functions in complex models. The ABC algorithm based on composite score functions (ABC-cs) searches for parameter values of the model of interest that produce simulated data which lead to composite score values -- at the observed maximum composite likelihood estimate -- close to those based on the original data. {The composite score statistic is rescaled with the corresponding information in order to take into account the amount of information on different parameter components. This rescaling has also the effect of making the ABC-cs procedure invariant to reparameterisations.

Although ABC-cs is not universally applicable, when it can be employed, e.g.  when  sensible composite likelihoods can be defined for the given model, it has several useful features. First of all, the summary statistic has dimension equal to the number of parameters, and it inherits, by construction, structural  stochastic characteristics of the model. Such statistic is also generally easy to compute, being often available analytically. 
Moreover, ABC-cs automatically adjusts the curvature of the composite likelihood and of the corresponding posterior distribution. Indeed, composite likelihoods typically do not satisfy the second Bartlett identity, also known as the information identity, and this usually leads to overly concentrated posterior distributions \citep{smith2009extended,pauli2011composite}. 
Hence, the straightforward use of the composite likelihood as a replacement to the full likelihood in Bayes' formula does not generally give a valid posterior distribution.
For this reason,  \citet{pauli2011composite} propose to first rescale the composite likelihood at the mode and then use this calibrated version in Bayes' theorem. This certainly improves inference, but sometimes may lead to the opposite problem of overestimating the variability in the posterior. From this point of view, at least in the examples considered here, the ABC-cs method gives better results, although computationally it may be more demanding, as is often the case with ABC methods.

There have been other attempts to merge composite likelihoods with the ABC framework.} {For instance, 
%\cite{erhardt2012approximate}, in the context of spatial extremes, combine composite likelihoods with ABC and show that this approach tends to work better than other existing methods. 
\cite{mengersen2013bayesian} use the composite score function with the empirical likelihood to produce an approximate  and weighted posterior sample. Their approach  is not ABC in the usual common sense, as it does not simulate from the full model. 
%Moreover, \cite{mengersen2013bayesian} do not address the adjustment issue of the composite likelihood. 
Also \citet[Sec. 7.1]{barthelme2011} mention the use of composite likelihoods in order to reduce the computational complexity of ABC, but they do not use the composite score as  summary statistic.}  

Our approach is similar in spirit to indirect inference  \citep[][]{heggland2004,gourieroux1993indirect}, as also the ABC-cs method  relies on an auxiliary model likelihood, that is the composite likelihood. 
%As happens in indirect inference, the closer the auxiliary model to the full model the more efficient the parameter estimates will be. 
%However, unlike indirect inference methods, {where the auxiliary model is chosen somehow subjectively, in our approach such auxiliary model is build upon composite likelihood theory}. 
{However,} ABC-cs is less computationally demanding since it does not require repeated maximisation for each simulated dataset. The indirect inference method within ABC has been discussed by \cite{drovandi2011approximate}. {More recently, also \cite{gleim2013} and \cite{drovandi2014} advocate the use of score functions based on auxiliary models as ABC summary statistics.

In Section 2 some background on ABC and  composite likelihood methods is given. The proposed ABC-cs algorithm is presented in Section 3. Section 4 illustrates the method {by} two examples, while Section 5 gives an application to spatial extreme data. Section 6 concludes the paper.

%%%%%%%%%%%%%%%%%%%%%%%%%%%%%%%%%%%%%%%%%%%%%%%%%%%%%%%%%%%%%%%%
\section{Statistical methods}
%%%%%%%%%%%%%%%%%%%%%%%%%%%%%%%%%%%%%%%%%%%%%%%%%%%%%%%%%%%%%%%%

%%%%%%%%%%%%%%%%%%%%%%%%%%%%%%%%%%%%%%%%%%%%%%%%%%%%%%%%%%%%%%%%
\subsection{ABC algorithms}
%%%%%%%%%%%%%%%%%%%%%%%%%%%%%%%%%%%%%%%%%%%%%%%%%%%%%%%%%%%%%%%%
Let $\pi(\theta)$ be a prior distribution for the parameter $\theta\in\Theta \subseteq \Real^d$, $L(\theta)=L(\theta;y)=f(y;\theta)$ the likelihood function based on data  $y$ and $\pi(\theta|y) \propto \pi(\theta) L(\theta)$ the posterior distribution of $\theta$. Suppose that $L(\theta)$ is unavailable for mathematical or computational reasons. 

The primary purpose of ABC algorithms is to approximate the posterior distribution when other methods, such as Markov chain Monte Carlo (MCMC), data augmentation, importance sampling or Laplace approximation cannot be used, but when the data from $f(y;\theta)$ can be easily simulated. Let $\eta(\cdot)$ be a set of summary statistics, $\rho(\cdot,\cdot)$ a distance function and $\epsilon>0$ a tolerance threshold. Moreover, let $y^{\mathrm{obs}}$ be the observed data. Then the ABC accept-reject algorithm (Algorithm~\ref{alg:abc})  

\vspace{1em}
\IncMargin{1em}
\begin{algorithm}[H]
  \SetAlgoLined
  \KwResult{A sample $(\theta^{(1)},\ldots,\theta^{(m)})$ from $\pi_\epsilon(\theta|\eta(y^{\text{obs}}))$}
  \For {$i = 1 \to m$}{
%  \Repeat{$z = y$}{
  \Repeat{$\rho(\eta(y), \eta(y^{\mathrm{obs}}))\leq\epsilon$}{
    \nl draw $\theta^*\,\sim\,\pi(\theta)$\\
    \nl draw $y\,\sim\, f(y;\theta^*)$\\
    %\nl{\If {$y\, =\, z$} {
}
 \nl set $\theta^{(i)}\, =\, \theta^*$
}
  \caption[]{\label{alg:abc} ABC accept-reject sampler.}
\end{algorithm}\DecMargin{1.em}
\vspace{1em}
\noindent samples from the  joint distribution
\begin{equation}
\pi_{\epsilon}(\theta,y|\eta(y^{\text{obs}})) = \frac{\pi(\theta)f(y;\theta)\mathbb{I}_{A_{\epsilon,y^{\text{obs}}}}(y)}{\int_{A_{\epsilon,y^{\text{obs}}}\times\Theta}\pi(\theta)f(y;\theta)\,dyd\theta},
\end{equation}
where $\mathbb{I}_{A_{\epsilon,y^{\text{obs}}}}(y)$ is the indicator function of the set  
$
A_{\epsilon,y^{\text{obs}}}(y) = \{y:\,\rho(\eta(y),\eta(y^{\text{obs}}))\leq\epsilon\}
$,
and it produces an approximation to the posterior distribution $\pi(\theta|y^{\text{obs}})$, given by
\[
\pi_{\epsilon}(\theta|\eta(y^{\text{obs}})) = \int \pi_{\epsilon}(\theta,y|\eta(y^{\text{obs}}))\,dy\,.%\approx\pi(\theta|\eta(y)).
\]  
If  $\epsilon\to 0$, then $\pi_{\epsilon}(\theta|\eta(y^{\text{obs}})) \to \pi(\theta|\eta(y^{\text{obs}}))$ \citep{blum2010approximate}. In addition, if $\eta(\cdot)$ is sufficient, then  $\pi_{\epsilon}(\theta|\eta(y^{\text{obs}})) \to \pi(\theta|y^{\text{obs}})$ \citep[see, for instance,][]{marin2012}.

The threshold $\epsilon$ cannot be fixed to zero, for  computational efficiency, and is generally set to the $\alpha$th quantile of the distance among the statistics, with $\alpha$ {being} typically very small \citep[see e.g.][]{beaumont2002}. With non-informative priors, the original accept-reject algorithm may be very inefficient
\citep{marin2012}. Nevertheless, this issue can be effectively addressed by using more advanced Monte Carlo algorithms, such as MCMC methods \citep{marjoram2003markov}, importance sampling \citep{fearnhead:2012to}, sequential or population Monte Carlo approaches \citep{sisson07,sisson09,beaumont2009adaptive,drovandi2011estimation,del2012adaptive}. Hence, the choice of $\eta(\cdot)$ is a crucial point of ABC. Indeed, what ABC can achieve at best is $\pi(\theta|\eta(y^{\text{obs}}))$, since $\eta(\cdot)$ is rarely sufficient. This loss of information seems to be a necessary price to pay for the access to computable quantities. {The idea here is to base the construction of $\eta(\cdot)$ on the score function of a composite likelihood, which is briefly recalled in the next section.}
 
%%%%%%%%%%%%%%%%%%%%%%%%%%%%%%%%%%%%%%%%%%%%%%%%%%%%%%%%%%%%%%%%
\subsection{Composite likelihoods}
%%%%%%%%%%%%%%%%%%%%%%%%%%%%%%%%%%%%%%%%%%%%%%%%%%%%%%%%%%%%%%%%
\label{compositelikelihoods}
%In various modern applications likelihood-based methods may encounter computational problems, due to the difficulty, or even impossibility, of evaluating the full likelihood function. In these situations, it is possible to resort to pseudo-likelihoods, called { composite likelihoods}, which are based on the composition of suitable lower dimensional densities, such as marginal or conditional densities or even a combination of them. 

Let $y=(y_1,\ldots,y_n)$ be a realisation of $Y=(Y_1,\ldots,Y_n)$, with independent components $Y_i\sim f(y_i;\theta)$, where $y_i \in \mathcal{Y} \subseteq \Real^q$, and let $\{A_1,\ldots,A_K\}$ be a set of marginal or conditional events on $\mathcal{Y}$. {The composite log-likelihood is defined as \citep[see, e.g.,][]{lindsay1988composite}
\begin{equation}
c\ell(\theta;y) = \sum_{i=1}^n\sum_{k=1}^{K} w_{k} \log f(y \in A_k;\theta),
\label{cl}
\end{equation}}
where $w_k$, $k = 1,\ldots,K$, are non-negative weights. When the events $A_k$  are defined in terms of pairs of bivariate marginal densities $f_{hk} (\cdot,\cdot;\theta)$, {then (\ref{cl}) is called the pairwise log-likelihood and is given by}
\begin{equation}
p\ell(\theta; y) = \sum_{i=1}^{n} \sum_{\substack{h,k=1\\ h\neq k}}^{q} w_{hk} \log f_{hk} (y_{ih},y_{ik};\theta).
\label{pl}
\end{equation}
%In some circumstances, it may be useful to consider larger subsets, such as triplets of observations, as in Example 4.
%The pairwise likelihood is a particular instance of the general class of composite likelihoods \citep[see][]{varin,varin2011overview,larribe2011composite}.

%The class of composite likelihoods contains, and thus generalizes, the usual ordinary likelihood, as well as many other alternatives, such as the pseudo-likelihood of \cite{besag1974spatial} and the Cox's partial likelihood \citep{cox1975partial}.

%In general, $c\ell(\theta;y)$ may be obtained from an optimal combination of events $\mathcal{S}_k$ of several types \citep{cox2004note}. The best choice of subsets of likelihoods is related to the separate topic of optimal design for likelihood estimation. This topic is, however, outside the scope of the present work.

The validity of inference about $\theta$ based on a composite likelihood can be assessed from the standpoint of unbiased estimating functions or the Kullback-Leibler criterion \citep{lindsay1988composite,cox2004note,lindsay2011issues,varin2011overview}.  Under rather broad assumptions  \citep[see, for instance,][]{molenberghs2005models}, the maximum composite likelihood estimator (MCLE) $\tilde \theta$ is the solution of the composite score equation
\begin{equation}
 c\ell_\theta (\theta;y) = \frac{\partial c\ell(\theta;y)}{\partial \theta}=0.
 \label{scorep}
\end{equation}
The composite score $c\ell_\theta (\theta;y)$ is unbiased, i.e. $E_\theta\{c\ell_\theta (\theta;Y)\} = 0$, since it is a linear combination of valid score functions. Moreover, $\tilde\theta$ is consistent and approximately normal, with mean $\theta$ and variance 
$$
V(\theta) =  H(\theta)^{-1} J(\theta) H(\theta)^{-1}
\ ,
$$
where $H(\theta)=E_\theta\{-\partial c\ell_\theta (\theta;Y)/\partial \theta^{\T}\}$ and $J(\theta)=\mbox{var}_\theta\{c\ell_\theta (\theta;Y)\}$ are the sensitivity and the variability matrices, respectively. For a full likelihood, $H(\theta) = J(\theta)$ and this is known as the information identity. The matrix $G(\theta)=V(\theta)^{-1}$ is known as the Godambe information, and the sandwich form of $V(\theta)$ is due to the failure of the information identity since, in general, $H(\theta) \neq J(\theta)$. This failure typically implies that the composite likelihood is wrongly too concentrated.

\cite{smith2009extended} discuss the use of the composite likelihood in  Bayes' theorem and notice that it may give overly too concentrated posteriors. \cite{pauli2011composite} suggest to combine   a calibrated composite likelihood $cL_c(\theta;y)=\exp\{ c\ell_c(\theta;y) \}$ with a prior $\pi(\theta)$ in order to obtain a calibrated composite posterior
\begin{equation}
\pi_c(\theta|y) \propto \pi(\theta) cL_c(\theta;y)
\ .
\label{postc}
\end{equation}
The calibrated composite likelihood is given by
\begin{equation}
cL_c(\theta;y) = cL(\theta;y)^{1/\bar{\omega}},
\label{plc}
\end{equation}
where $\bar\omega = \mbox{trace}\{J(\tilde\theta)H(\tilde\theta)^{-1}\}/d$.  This calibration approximately adjusts the curvature of the composite likelihood and allows to recover the asymptotic properties of a posterior distribution. Examples of (\ref{postc}) are discussed in \citet{pauli2011composite}; see also \cite{ribatet.etal2012} for other types of adjustments. 

Bayesian inference based on composite likelihoods leads to composite posteriors, which depend crucially on the calibration adjustment in (\ref{plc}). Since $\bar\omega$ is evaluated at $\tilde\theta$, this  calibration might lead to composite posteriors \eqref{postc} with overestimated variability (see Section 4).

%%%%%%%%%%%%%%%%%%%%%%%%%%%%%%%%%%%%%%%%%%%%%%%%%%%%%%%%%%%%%%%%
\section{ABC with unbiased estimating functions}
%%%%%%%%%%%%%%%%%%%%%%%%%%%%%%%%%%%%%%%%%%%%%%%%%%%%%%%%%%%%%%%%

%{\cambios REF 1: questo paragrafo va riscritto}

%In the ABC context, the similarity of simulated and observed data is typically measured by means of a distance between some summary statistics, which are in general not sufficient. {\cambios REF 2 punt 9: On the other hand, to reduce the approximation error, the summary statistics should be as low-dimensional as possible \citep{fearnhead:2012to}}. In general, the choice of the summary statistics is not straightforward, especially with high-dimensional data and complex model structures.
We propose a suitably rescaled composite score function -- evaluated at the observed MCLE -- as the summary statistic for ABC.
%We propose as the summary of the data in ABC a suitably rescaled composite score function $c\ell_\theta (\theta;y)$, evaluated at  the observed MCLE, $\tilde\theta^{\text{obs}}$.
 %the use in ABC of a suitably rescaled composite score function $c\ell_\theta (\theta;y)$, evaluated at  the observed MCLE, $\tilde\theta^{\text{obs}}$, as a summary of the data. 
This leads to the ABC-cs algorithm, which is introduced and  discussed in Section \ref{ABC-CS}. The aim of Section \ref{abs-score} is to provide a logical motivation for the proposal of Section \ref{ABC-CS},  by discussing the ideal, although unrealistic, situation in which a full computable likelihood is available.

%%%%%%%%%%%%%%%%%%%%%%%%%%%%%%%%%%%%%%%%%%%%%%%%%%%%%%%%%%%%%%%%
\subsection{ABC with score functions}\label{abs-score}
%%%%%%%%%%%%%%%%%%%%%%%%%%%%%%%%%%%%%%%%%%%%%%%%%%%%%%%%%%%%%%%%

%{\cambios REF 1 punto 1: va aggiunta una frase per spiegare perche' partiamo dalla famiglia esponenziale}

In this section we show how the score function evaluated at the observed maximum likelihood estimate provides a natural summary statistic for ABC in the, admittedly restrictive, case in which a full likelihood is available. In the following, we first start with a full exponential  model and then extend the reasoning to a generic model.

Consider a full exponential family with density 
\begin{eqnarray}
f(y;\varphi)=h(y) \exp\{\varphi^\T s(y) - k(\varphi)\}
\ ,
\label{dexp}
\end{eqnarray} 
where $h(y)>0$, $\varphi$ is the   canonical parameter, $s(y)$ is the $d$-dimensional sufficient statistic, and $k(\varphi)$ is the cumulant generating function of $s(y)$. In this case, the obvious summary statistic for ABC is the minimal sufficient statistic $s(y)$, which  gives the exact posterior for $\epsilon\to 0$  {\citep[see, e.g.,][]{blum2010approximate}.} %\citep[see, e.g.,][]{rubio2013simple}.
The following proposition shows that the ABC posterior based on a suitably rescaled score function is exact for $\epsilon\to 0$ and also invariant to reparameterisations.

\begin{proposition}
\label{p1}
Let $\ell(\varphi;y)=\varphi^\T s(y)-k(\varphi)$ be the log-likelihood for $\varphi$ based on model (\ref{dexp}), and consider as the summary statistic the rescaled score evaluated at a fixed $\varphi_0$, that is
$$
\eta(y;\varphi_0) = B(\varphi_0)^{-1} \ell_\varphi(\varphi_0;y)
\ ,
$$
where $\ell_\varphi(\varphi;y)=\partial \ell(\varphi;y)/\partial \varphi=s(y) - \partial k(\varphi)/\partial \varphi$  and $B(\varphi)$ is such that $i(\varphi)= \partial^2 k(\varphi)/(\partial \varphi \partial \varphi^{\T})= B(\varphi) B(\varphi)^{\T}$.  {Then, the ABC posterior based on $\eta(y;\varphi_0)$ is exact for $\epsilon\to 0$ and also invariant to reparameterisations, regardless of the fixed value $\varphi_0$.} 
\end{proposition}

\begin{proof}
For any fixed value $\varphi_0$, the rescaled score $\eta(y;\varphi_0)$ is a linear transformation of the minimal sufficient statistic $s(y)$, and thus it is itself minimal sufficient.  This proves that the ABC posterior based on $\eta(y;\varphi_0)$ is exact for $\epsilon\to 0$.

Consider the reparametrisation $\theta=\theta(\varphi)$. Let $\bar\ell(\theta)=\ell(\varphi(\theta))$ and $\bar{\imath}(\theta)=\varphi_{\theta}^{\T} i(\varphi(\theta)) \varphi_{\theta}$, where $\varphi_{\theta}=\partial \varphi(\theta)/\partial \theta$. The rescaled score is $\bar\eta(y;\theta_0) = \bar{B}(\theta_0)^{-1} \bar\ell_\theta(\theta_0;y)$, with $\theta_0=\theta(\varphi_0)$, $\bar\ell_\theta(\theta;y)=\partial \bar\ell(\theta;y)/\partial \theta$ and $\bar{B}(\theta)$ such that $\bar{B}(\theta) \bar{B}(\theta)^{\T}= \bar{\imath}(\theta)$. Then, since $\bar{B}(\theta)=\varphi_{\theta}^{\T} B(\varphi(\theta))$ and $\bar\ell_\theta(\theta;y)=\varphi_{\theta}^{\T} \ell_\varphi(\varphi(\theta);y)$, it follows that $\bar\eta(y;\theta_0) = \eta(y;\varphi_0)$. This proves invariance to reparameterisations.$\blacksquare$
\end{proof}

Proposition \ref{p1} holds for any value of $\varphi_0$. In particular, when $\varphi_0$ is the observed value of the  maximum likelihood estimate (MLE) at the observed data $y^{\text{obs}}$, i.e. $\hat\varphi^{\text{obs}}$, we have $\eta(y^{\text{obs}};\hat\varphi^{\text{obs}})=0$. This choice of $\varphi_0$ is particularly convenient for a general model $f(y;\theta)$. Indeed, in this case, at least in principle, we could use an alternative representation of $y$, or equivalently the minimal sufficient statistic based on $y$, given by $(\hat\theta,a)$, where $\hat\theta$ is the MLE and $a$ is an ancillary statistic, which means that its distribution does not depend on $\theta$. Hence, we could replace $f(y;\theta)$ with $f(\hat\theta,a;\theta)$, and the latter can be factorised as
$$
f(\hat\theta,a;\theta)=f(\hat\theta|a;\theta) f(a)\,.
$$
This means that the likelihood for $\theta$ can be based equivalently on $f(y;\theta)$ or $f(\hat\theta|a;\theta)$. Unfortunately, it may not be easy in general to find $f(\hat\theta|a;\theta)$. On the other hand,  it is possible to approximate such density through a {\em tangent exponential model} at (and near) the fixed value $y^{\text{obs}}$ \citep[][Sect. 3.2]{Fraser.Reid:1995,reid2003}. Denoting by $\ell(\theta;y^{\text{obs}})$ the observed log-likelihood, the approximation to the log-likelihood based on the tangent exponential model is 
\begin{equation}\label{TEM}
\ell^{\text{TE}}(\theta;y) = \ell(\theta;y^{\text{obs}})-\ell(\hat\theta^{\text{obs}};y^{\text{obs}})+\{\varphi(\theta)-\varphi(\hat\theta^{\text{obs}})\}^\T s(y)\,,
\end{equation}
where $\hat\theta^{\text{obs}}$ is the MLE at the observed data  point $y^{\text{obs}}$, $s(y)=\partial \ell(\theta; y)/\partial\theta|_{\theta=\hat\theta^{\text{obs}}}=\ell_\theta(\hat\theta^{\text{obs}};y)$,  and $\varphi(\theta)=\varphi(\theta;y^{\text{obs}})$ is a one-to-one reparameterisation dependent on the observed data $y^{\text{obs}}$ \citep[see also][Sect. 8.4.2]{bdr2007}. The tangent exponential model is a local exponential family model with sufficient statistic $s(y)$ and canonical parameter $\varphi$. It has the same log-likelihood function as the original model at the fixed point $y^{\text{obs}}$, where it also has the same first derivative with respect to $y$.

From Proposition \ref{p1}, the summary statistic for ABC for the tangent exponential model (\ref{TEM}) is the rescaled score, where the score is given by
\begin{eqnarray}
\ell^{\text{TE}}_\theta(\theta;y) = \ell_\theta(\theta;y^{\text{obs}})+\varphi_\theta s(y)
\ .
\label{tes}
\end{eqnarray}
For $\theta=\hat\theta^{\text{obs}}$, (\ref{tes}) reduces to $\varphi_\theta(\hat\theta^{\text{obs}}) \ell_\theta(\hat\theta^{\text{obs}};y)$, i.e.\ to a linear transformation of the score of the original model. Rescaling (\ref{tes}) then provides invariance to reparameterisation, as in Proposition \ref{p1}. This motivates the use of the score function evaluated at $\hat\theta^{\text{obs}}$ as an approximate optimal summary statistic in ABC for a general model. \\

\noindent {\bf Example 1: normal parabola}.  Let $y=(y_1,\ldots,y_n)$ be a random sample from the normal distribution $N(\theta,\theta^2)$, with $\theta>0$. The log-likelihood is
$$
\ell(\theta;y)=\frac{1}{\theta}\sum_{i=1}^n y_i -\frac{1}{2\theta^2} \sum_{i=1}^n y_i^2 - n\log\theta\,,
$$ 
where $t(y)=(\sum_{i=1}^n y_i,\sum_{i=1}^n y_i^2)$ is the two-dimensional minimal sufficient statistic. The
score function is $\ell_\theta(\theta;y)=-\theta^{-2} \sum_{i=1}^n y_i +\theta^{-3} \sum_{i=1}^n y_i^2-n\theta^{-1}$, which implies that $\hat\theta$ is the positive solution of a quadratic equation. The expected information is $i(\theta)=3n/\theta^2$, and the rescaled score is $\eta(y;\hat\theta^{\text{obs}})= \hat\theta^{\text{obs}} \ell_\theta(\hat\theta^{\text{obs}};y) /\sqrt{3n}$.

As an illustration we use a sample of size $n=50$ generated from the model, with $\theta=5$ and with a uniform prior in $(0,15)$. We consider three instances of the ABC Algorithm~\ref{alg:abc}, with distance $\rho(v,w)=||v-w||_1$ and with summary statistics given, respectively, by $t(y)$, $\eta(y;\hat\theta^{\text{obs}})$, and also a one-to-one transformation of the minimal sufficient
statistic $t(y)$, that is $t_1(y)=(\bar y,\sqrt{s^2})$, i.e. the sample mean and the standard deviation. In all three cases, we {use} the same sample of $10^7$ values generated from the prior and in each case we choose the threshold $\epsilon$ as the quantile of level $0.1\%$ of the observed distances, thus accepting $10^4$  values. These $\epsilon$ values are, respectively, $31.264$, $0.02$ and $0.237$. These values are not directly comparable, since the three statistics are not on the same scale. A possibility would be to suitably standardize $t(y)$ and $t_1(y)$, but such a standardisation is not obvious in general. On the other hand, the statistic $\eta(y;\hat\theta^{\text{obs}})$ is rescaled using the variability of the score. For vector parameters this rescaling also takes into account the correlation among the components of the statistic.

\begin{figure}[h!]
\centering
{\includegraphics[width=2.1in, height=6.6in, angle=-90]{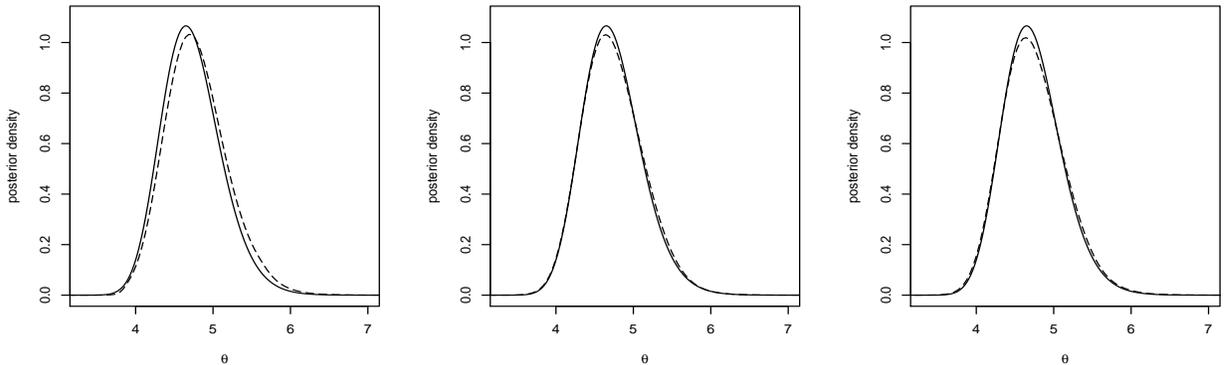}}
\caption{\small Normal parabola. In all panels the solid line corresponds to the exact posterior, while the dashed lines correspond to ABC approximations using $t(y)$ (left panel), $t_1(y)$ (central panel), and $\eta(y;\hat\theta^{\text{obs}})$ (right panel).} %Vertical line is the true parameter value.}
\label{normal-parabola}
\end{figure}

Figure \ref{normal-parabola}
shows the three approximations  compared with the exact posterior. The two versions of the ABC with the minimal sufficient statistic  gave quite different results, with the one with $t(y)$ leading to the worst accuracy. This is likely due to the fact that the the two components of $t(y)$ are on different scales. %large value of $\epsilon$ ($31.264$). %We note that only three of the $10^7$ proposed  values of $\theta$ would have been accepted with $\epsilon=1$, thus making the ABC algorithm with $t(y)$ impractical.
On the other hand, the ABC with the one-dimensional summary statistic $\eta(y;\hat\theta^{\text{obs}})$, which is not sufficient for this model, gave an approximation to the posterior with accuracy comparable with ABC with the minimal sufficient statistic $t_1(y)$. 

{In order to check that this behaviour is not due to the particular simulated dataset, we consider the same experiment on 50 different datasets, and for each posterior we compute the Kullback-Leibler (KL) divergence among the exact and the three approximate posteriors. A plot of the log-KL divergences is given in Figure~\ref{normal-parabola2}, which confirms the good agreement of ABC with $\eta(y;\hat\theta^{\text{obs}})$ and ABC with the minimal sufficient statistic $t_1(y)$, but not with the minimal sufficient statistic $t(y)$.

\begin{figure}[h!]
\centering
{\includegraphics[scale=0.7, angle=-90]{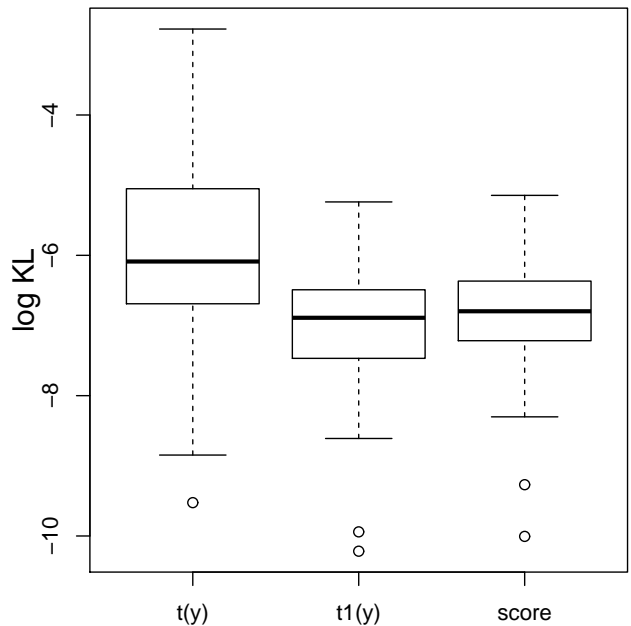}}
\caption{Kullback-Leibler divergences (logarithmic scale) among the exact and ABC posteriors using $t(y)$, $t_1(y)$ and $\eta(y;\hat\theta^{\text{obs}})$, over 50 replicated datasets for the normal parabola.}
\label{normal-parabola2}
\end{figure}
}
~\vspace{0.5cm}

\noindent {\em Remark 1.}  From the point of view of the likelihood principle, the different performances of the ABC algorithm  with {both} versions of the minimal sufficient statistic in Example 1 is  unpleasant. Indeed,
$t(y)$ and $t_1(y)$  lead to the same likelihood function and posterior distribution, but the corresponding ABC approximations could be remarkably different. {Hence, transforming the summary statistic may have a great impact on the quality of the ABC approximation. Finding the right transformation may not be straightforward, especially when the summary statistic is high-dimensional. This issue has already been recognised in the ABC literature. For instance, \cite{HsuanJung2011ga} propose to weight the components of the summary statistic using  a genetic algorithm, though the method seems computationally quite intensive}. On the contrary, since the likelihood and the score functions are not affected by one-to-one transformations of the data, or of the minimal sufficient statistic, ABC with $\eta(y;\hat\theta^{\text{obs}})$ is invariant with respect to such transformations. This  invariance to data transformations adds to  the parameterisation invariance proved in Proposition \ref{p1}.
~\vspace{0.5cm}

\noindent {\em Remark 2.} Although the choice of the distance function $\rho(\cdot,\cdot)$ in the ABC algorithm  is arbitrary, when considering the Euclidean distance we have
\begin{eqnarray*}
\rho \left( \eta(y;\hat\theta^{\text{obs}}),\eta(y^{\text{obs}};\hat\theta^{\text{obs}})\right) & = &
|| \eta(y;\hat\theta^{\text{obs}})||_2^{1/2} = \left\{ \ell_\theta (\hat\theta^{\text{obs}};y)^{\T} i(\hat\theta^{\text{obs}})^{-1} \ell_\theta(\hat\theta^{\text{obs}};y) \right\}^{1/2}
\ ,
\end{eqnarray*}
which is {the square root of} the score test statistic computed in $\hat\theta^{\text{obs}}$, based on data $y$.
~\vspace{.5cm}

Despite the good properties of ABC with the score function,  in typical applications of the ABC method the likelihood function is intractable, and therefore the same is true for the score function. This motivates the extension to composite likelihoods proposed in the next section.

%%%%%%%%%%%%%%%%%%%%%%%%%%%%%%%%%%%%%%%%%%%%%%%%%%%%%%%%%%%%%%%%
\subsection{ABC with composite score function}\label{ABC-CS}
%%%%%%%%%%%%%%%%%%%%%%%%%%%%%%%%%%%%%%%%%%%%%%%%%%%%%%%%%%%%%%%%

When dealing with complex models, possible surrogates of the unavailable full likelihood are given by composite likelihoods. Extending the results of the previous section, we propose the rescaled composite score function as a summary statistic in ABC. This defines an algorithm, called ABC-cs. In terms of the ABC Algorithm~\ref{alg:abc}, ABC-cs replaces the matching condition
$$
\rho(\eta(y),\eta(y^{\text{obs}}))\leq\epsilon
 ,
$$
with
\begin{eqnarray}
\rho \left( \eta_c (\tilde\theta^{\text{obs}};y),\eta_c(\tilde\theta^{\text{obs}};y^{\text{obs}}) \right) \leq\epsilon
\ ,
\label{distc}
\end{eqnarray} 
where $\tilde\theta^{\text{obs}}$ is the MCLE  computed with $y^{\text{obs}}$ and 
\begin{eqnarray}
\eta_c (\tilde\theta^{\text{obs}};y) = B_c(\tilde\theta^{\text{obs}})^{-1} c\ell_\theta(\tilde\theta^{\text{obs}};y)
\ 
\label{etac}
\end{eqnarray}
is the rescaled composite score, with $B_c(\theta)$ such that $J(\theta)=B_c(\theta) B_c(\theta)^{\T}$. Since  $c\ell_\theta(\tilde{\theta}^{\text{obs}};y^{\text{obs}})=0$,  in (\ref{distc}) we only need to evaluate $\eta_c(\tilde{\theta}^{\text{obs}};y)$.

The following theorem shows that the proposed ABC-cs algorithm gives an approximate posterior distribution with the correct curvature, in the sense discussed at the end of Section \ref{compositelikelihoods}, even if the rescaled composite score function (\ref{etac}), unlike the full score function, does not satisfy the information identity.

\begin{theorem}
The ABC-cs algorithm with the rescaled composite score statistic $\eta_c (\tilde\theta^{\text{obs}};y),$ as $\varepsilon \to 0$, leads to an approximate posterior distribution with the correct curvature and is also invariant to reparameterisations.
\end{theorem}

\begin{proof}
In order to recover the information identity, and thus the correct curvature, it is necessary to consider the adjusted composite score function \citep[see, e.g.,][Chap. 4]{pacesalvan97}
$$
g(\theta;y)=H(\theta)J(\theta)^{-1} c\ell_\theta(\theta;y)=A(\theta)c\ell_\theta(\theta;y)
\ .
$$
 Indeed,
for $g(\theta;y)$ we have 
\begin{eqnarray*}
J_g(\theta) \,=\, \mbox{var}_\theta\{g(\theta;Y)\} = A(\theta)\mbox{var}_\theta\{c\ell_\theta(\theta;Y)\}A(\theta)^\T=G(\theta)
\end{eqnarray*}
and, using $E_\theta\{c\ell_\theta(\theta;Y)\}=0$,
\begin{eqnarray*}
H_g(\theta) \, &=&\, E_\theta\left\{-\frac{\partial}{\partial\theta^\T}g(\theta;Y)\right\} = %\\
%& = & -\left\{\frac{\partial}{\partial\theta^T}A(\theta) \right\} E_\theta\{c\ell_\theta(\theta;Y)\} 
- A(\theta)E_\theta\left\{\frac{\partial}{\partial\theta^\T}c\ell_\theta(\theta;Y)\right\}=G(\theta)\,.% \\
%& = & G(\theta).
\end{eqnarray*}
Since $H_g(\theta)=J_g(\theta)=G(\theta)$, the adjusted composite score $g(\theta;y)$ satisfies the information identity as a proper score function and, since {$|A(\theta)| \neq 0$,  $g(\theta;y)=0$} leads to {the same solution $\tilde\theta$ of the estimating equation} $c\ell_\theta(\theta;y)=0$.

The ABC-cs algorithm should then be based on the rescaled version of $g(\theta;y)$, given by
$$
\eta_g (\tilde\theta^{\text{obs}};y) = B_g (\tilde\theta^{\text{obs}})^{-1} g(\tilde\theta^{\text{obs}};y)
\ ,
$$
where $B_g(\theta)=H(\theta) \{B_c(\theta)^{\T}\}^{-1}$. Indeed,
$$
G(\theta)= H(\theta) J(\theta)^{-1} H(\theta) = H(\theta) \{B_c(\theta)^{\T}\}^{-1} B_c(\theta)^{-1} H(\theta)
\ .
$$
However, it is straightforward to see that 
$$
\eta_g (\theta;y) = B_g (\theta)^{-1} g(\theta;y) = B_c(\theta)^{\T} H(\theta)^{-1} H(\theta) J(\theta)^{-1} c\ell_\theta(\theta;y) = \eta_c(\theta;y)
\ .
$$
This proves that the use of $\eta_c(\tilde\theta^{\text{obs}};y)$ as a summary statistic for ABC leads to an approximate posterior with the correct curvature.  

The proof of invariance to reparameterisation follows the same steps as in Proposition 3.1.$\blacksquare$
\end{proof}

An advantage of ABC-cs is that the rescaled composite score statistic has the same dimension as $\theta$.  Moreover, since the score statistic is obtained from the composite log-likelihood by just taking the first derivative, it is easily computed, especially when it is analytically available.  An apparent drawback of (\ref{etac}) is the implicit dependence of the ABC-cs algorithm on $J(\theta)$. However, only $J(\tilde\theta^{\text{obs}})$ is needed, and this quantity can be easily approximated with a preliminary Monte Carlo simulation from the model with $\theta = \tilde\theta^{\text{obs}}$, with {few hundred replications} \citep{cattelan14}. Finally, note that  even in this case, the squared Euclidean distance gives the composite score test statistic evaluated in $\tilde\theta^{\text{obs}}$, based on data $y$.

The ABC-cs algorithm delivers an approximate posterior distribution which does not need calibration, whereas Bayesian composite posteriors depend crucially on such quantities. Moreover, even when rescaled, the Bayesian composite posterior (\ref{postc}) often leads to less accurate results, as also shown in the examples of Section 4  and in the application of Section 5.

As a final remark, we note that the proposal of this paper is not providing  an automatic  summary statistic for ABC, in the sense that an appropriate choice of composite likelihood for the problem under investigation must be made. The composite likelihood may be difficult, if not impossible, to define  in some applications \citep[see, for instance,  the non-Markovian queueing model analysed by][]{heggland2004}, while in other situations there could be different competing composite likelihoods available for the same model. The latter case will be addressed more in detail in the final discussion. The point here is that, when there is at least one composite likelihood available, it is usually defined starting from relevant stochastic features of the model and therefore the summary statistic based on the composite score will automatically incorporate these features. Moreover, there is an extensive, and growing, frequentist literature on composite likelihoods \citep[see, for instance, the review by][]{varin2011overview}, that can be used to guide the choice of a sensible composite likelihood in specific classes of models.

%%%%%%%%%%%%%%%%%%%%%%%%%%%%%%%%%%%%%%%%%%%%%%%%%%%%%%%%%%%%%%%%
\section{Examples}
%%%%%%%%%%%%%%%%%%%%%%%%%%%%%%%%%%%%%%%%%%%%%%%%%%%%%%%%%%%%%%%%

In the  examples below we use composite marginal likelihood functions \citep{cox2004note}, although different model structures might lead to different choices of suitable composite likelihoods.
We use the Godambe information $G(\tilde\theta^{\text{obs}})$ as a precision matrix for both  ABC and ABC-cs with importance sampling. Note that ABC with MCMC or Sequential Monte Carlo (SMC) methods requires a similar precision matrix, which in practice is estimated by considering preliminary runs of ABC (in the case of MCMC) or from a previous population of ABC particles (in the case of SMC). The \texttt{R} code for the examples of this section and for the application in Section~5 can be found in the Supplementary Material.\\

%%%%%%%%%%%%%%%%%%%%%%%%%%%%%%%%%%%%%%%%%%%%%%%%%%%%%%%%%%%%%%%%
\noindent {\bf Example 2: equi-correlated normal model}\\
%%%%%%%%%%%%%%%%%%%%%%%%%%%%%%%%%%%%%%%%%%%%%%%%%%%%%%%%%%%%%%%%
This  example %, considered in \cite{pace2011adjusting} among others, 
focuses on  inference based on the pairwise log-likelihood (\ref{pl}) for the parameters of an equi-correlated multivariate normal distribution, with mean vector $\mu$ and covariance matrix $\Sigma_{rs}=\rho\sigma^2$, for $r\neq s$, and $\Sigma_{rr}=\sigma^2$, $r,s = 1,\ldots,q$. For this model, $\tilde\theta$ is {fully} efficient, the sufficient statistic is three-dimensional and is the same for both the full and pairwise likelihoods \citep{pace2011adjusting}. The pairwise log-likelihood (\ref{pl}) for $\theta=(\mu,\sigma^2,\rho)$ is
\begin{align*}
p\ell(\theta;y)=-&\frac{nq(q-1)}{2}\log \sigma^{2} - \frac{nq(q-1)}{4}\log(1-\rho^{2}) - \frac{q-1+\rho}{2\sigma^{2}(1-\rho^{2})}SS_{W}\\
-&\frac{q(q-1)SS_{B}+nq(q-1)(\bar y -\mu)^{2}}{2\sigma^{2}(1+\rho)},
\end{align*}
where $SS_W = \sum_{i=1}^n \sum_{r=1}^q (y_{ir}-\bar y_i)^2$, $SS_B=\sum_{i=1}^n(\bar y_i - \bar y)^2$, $\bar y_{i}=\sum_{r=1}^{q} y_{ir}/q$ and $\bar{y} = \sum_{i=1}^{n}\sum_{r=1}^{q}y_{ir}/(nq)$. For the expression of the score function see \citet[p. 145]{pace2011adjusting}.

We assume that the components of the parameter $\omega = (\mu,\tau,\kappa)$, with $\tau = \log \sigma^2$ and  $\kappa = \text{logit}(\{\rho(q-1)+1\}/q)$, are independent, with $N(0,100)$ marginal prior distributions.
 
As an illustration, we use a sample of $n=30$ drawn from the model with $q=50$, $\mu=0$, $\sigma^2=1$ and $\rho=0.5$. For ABC we used the sufficient statistic $(\bar{y}, \sqrt{SS_B}, \sqrt{SS_w})$, which gave better results than the alternative form $(\bar{y}, {SS_B}, {SS_w})$, while for ABC-cs the summary statistic is given by \eqref{etac}. The simulation from the ABC and ABC-cs posteriors is performed with importance sampling, where the importance function is the multivariate $t$-student distribution with 5 degrees of freedom, centred at $\tilde{\theta}^{\mathrm{obs}}$ and with scale matrix equal to $5V(\tilde\theta^{\text{obs}})$. We consider $10^3$ final samples obtained with $\epsilon$ fixed at the $0.1\%$ quantile of the observed distances. Finally, in order to get rid of the importance weights, here and elsewhere, we consider resampling with replacement of the simulated values.

Results are compared also with the pairwise posterior
\begin{equation}
\pi_{pl} (\theta|y) \propto \pi(\theta) \exp\{p\ell(\theta;y)\}
\ ,
\label{plnc}
\end{equation}
with the pairwise posterior \eqref{postc} based on the calibrated pairwise likelihood and with the posterior distribution based on the full likelihood, approximated by a random walk Metropolis.

The boxplots of the marginal posterior approximations are shown in Figure~\ref{fig42-1}, which  highlights several interesting features. The posterior \eqref{plnc} appears wrongly too concentrated \citep[see also][]{pauli2011composite,smith2009extended,ribatet.etal2012}, whereas the calibrated pairwise posterior \eqref{postc} may have the  opposite problem. Indeed, while the marginal calibrated pairwise posteriors of $\mu$ and $\tau$ are quite similar to the full posterior (MCMC), the marginal calibrated pairwise posterior of $\kappa$ shows higher dispersion than the corresponding marginal based on the full likelihood.
\begin{figure}[h!]
\centering
{\includegraphics[width=2in, height=2.15in, angle=-90]{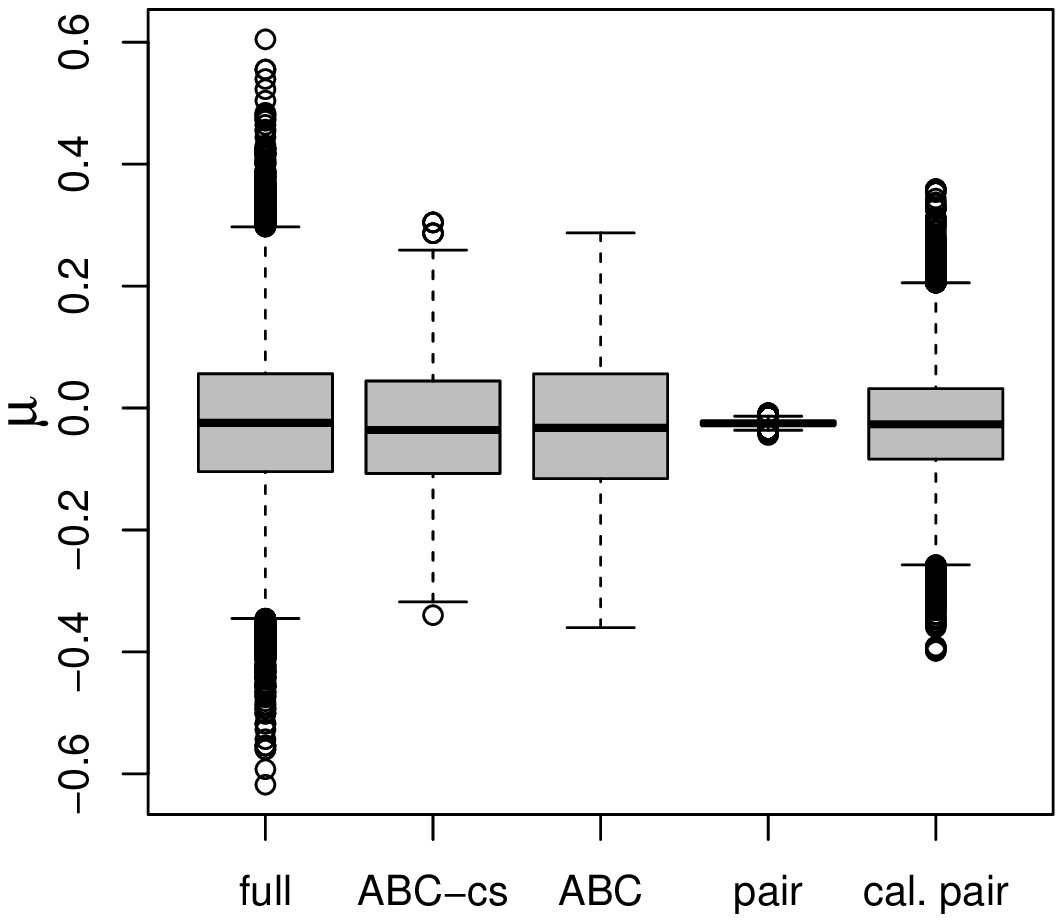}}
{\includegraphics[width=2in, height=2.1in, angle=-90]{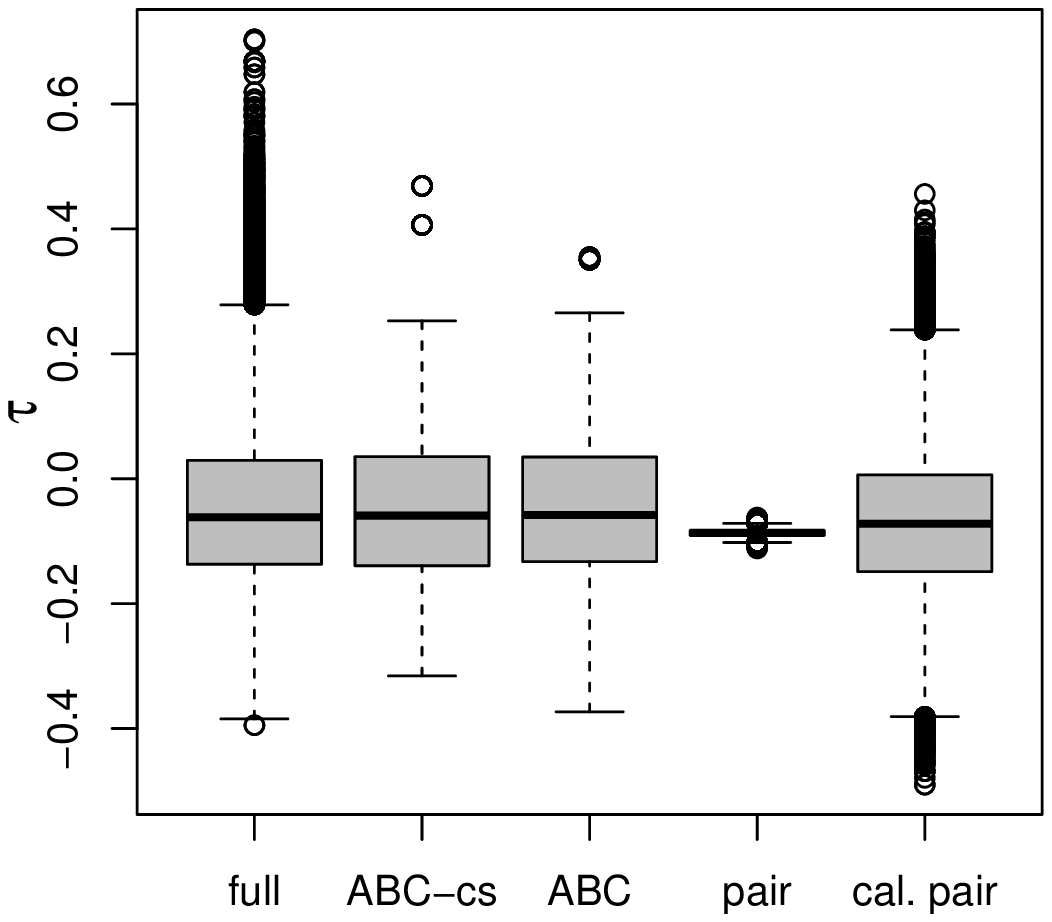}}
{\includegraphics[width=2in, height=2.15in, angle=-90]{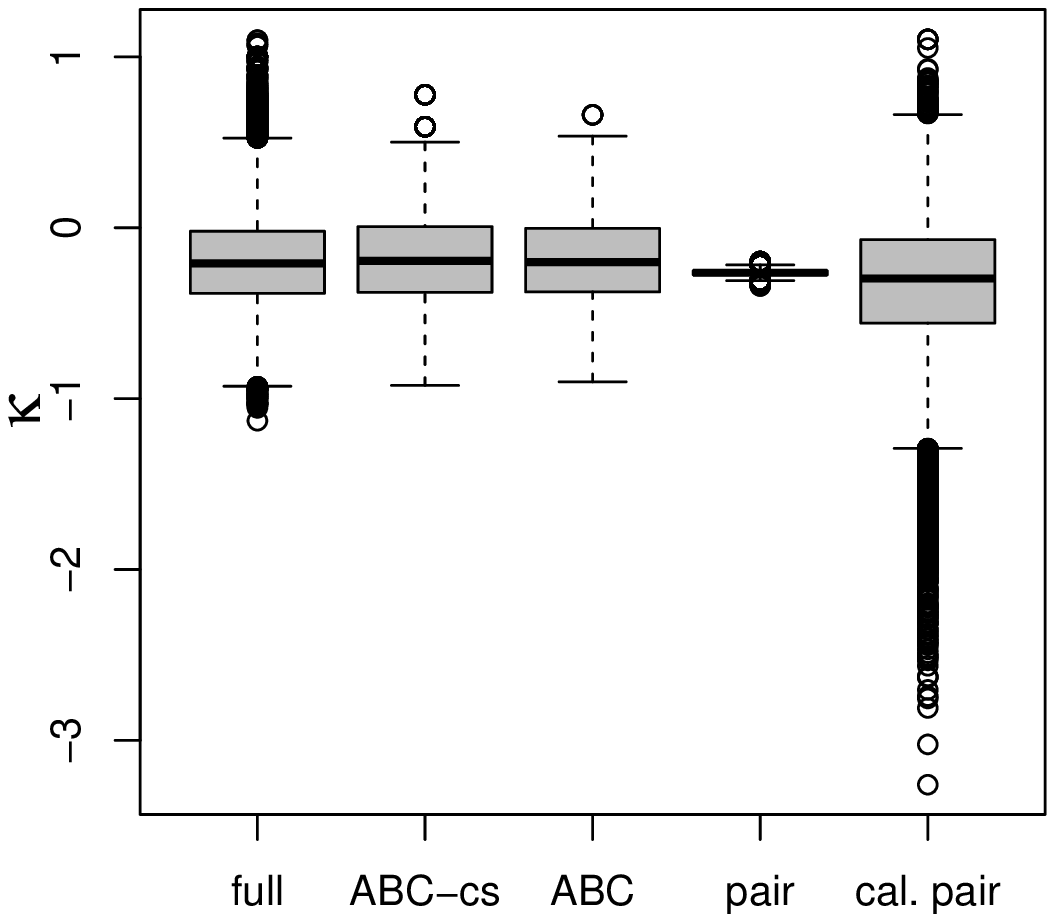}}
\caption{\small Equi-correlated normal model. ABC-cs posterior compared with
       the full, the pairwise (pair), the calibrated pairwise (cal. pair) and the ABC posteriors.}\label{fig42-1}
\end{figure}
On the other hand, ABC-cs and ABC marginal posteriors are all quite similar to the full posterior. This is not surprising, since the model is a full exponential family of order three and ABC uses exactly the sufficient statistic as summary statistic. Moreover, even the pairwise likelihood has exponential form, with the same sufficient statistic. This implies that the pairwise score function is  proportional to the score function of the full model \citep{pagui2015full} and the latter would lead again to the sufficient statistic (see Section \ref{abs-score}).

We also compare the posterior means of the full, ABC,  ABC-cs and the calibrated pairwise posteriors in a simulation study, over 100 Monte Carlo trials. The data are generated from the model with $\mu=0$, $\sigma^2=1$, $\rho=0.2$.  Figure~\ref{fig42:2} indicates that ABC and ABC-cs posterior means are quite similar to the full posterior mean, as expected from Proposition 3.1. On the contrary, for the transformed correlation parameter $\kappa$ the mean of the calibrated pairwise posterior can perform poorly.  The behavior of the calibrated pairwise likelihood is due to the fact that the overall rescaling, computed at the mode, does not generally guarantee accuracy in the tails. Simulations for other parameter configurations can be found in the Supplementary Material.

\begin{figure}[h!]
\centering
{\includegraphics[width=2.1in, height=2.15in, angle=-90]{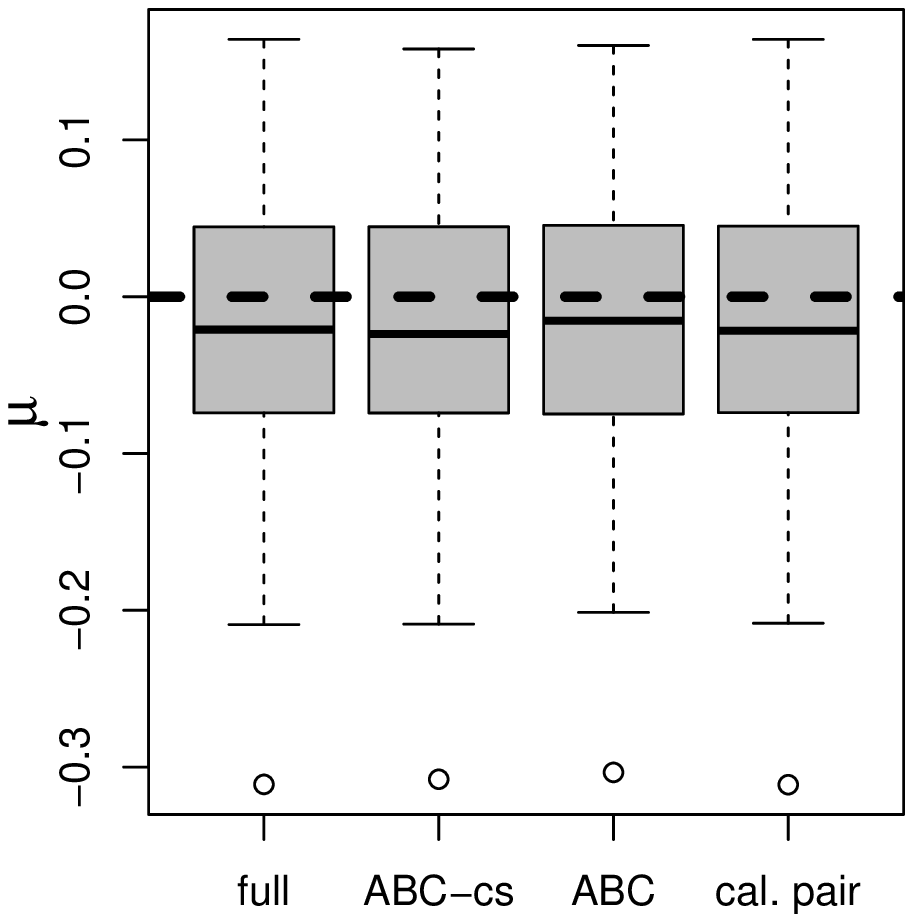}
\includegraphics[width=2.1in, height=2.1in, angle=-90]{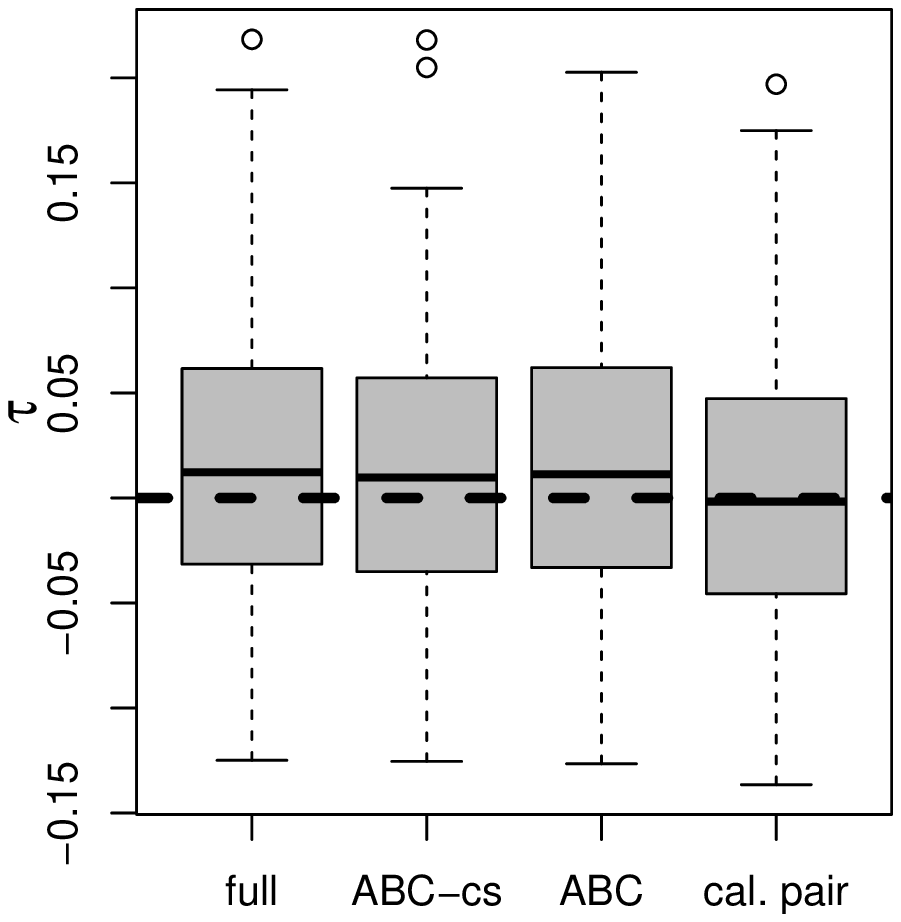}
\includegraphics[width=2.1in, height=2.15in, angle=-90]{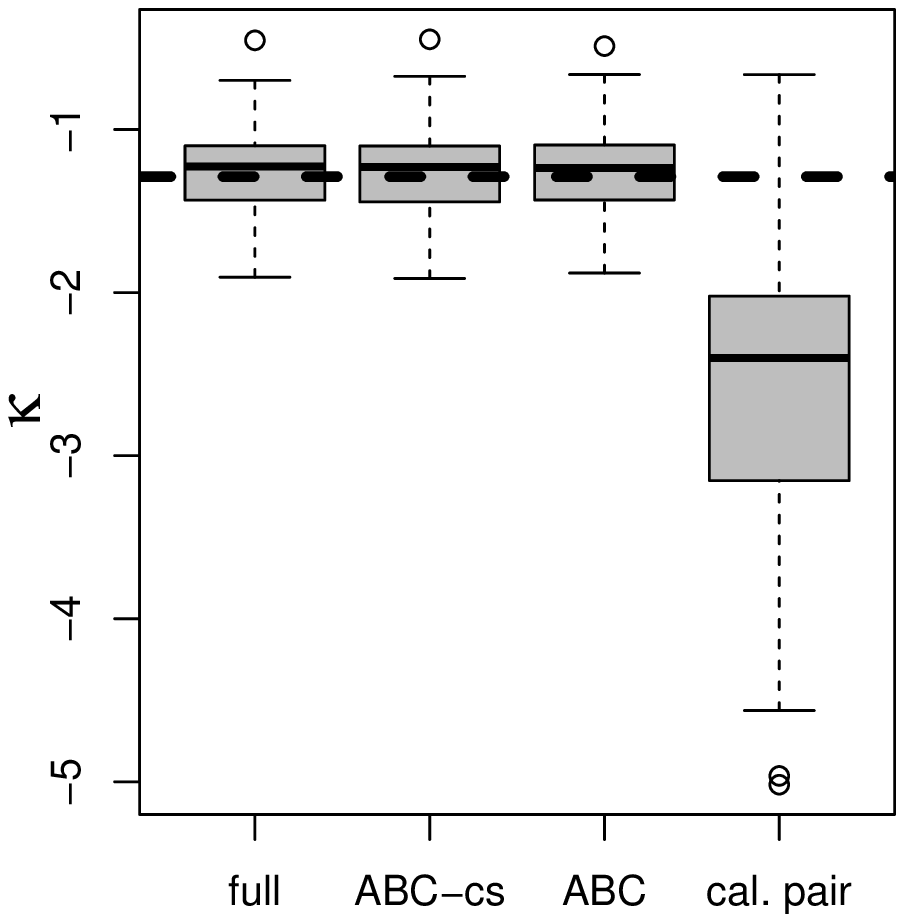}}\\
\caption{\small Equi-correlated normal model. Simulation study based on 100 Monte Carlo trials, with $\mu=0$, $\sigma=1$ ($\tau=0$) and $\rho=0.2$ ($\kappa \approx -1.15$). The dashed horizontal lines represent the true parameter values. }
\label{fig42:2}
\end{figure}
\hspace{0.3cm}

%%%%%%%%%%%%%%%%%%%%%%%%%%%%%%%%%%%%%%%%%%%%%%%%%%%%%%%%%%%%%%%%
\noindent{\bf Example 3: multivariate probit model} \\
%%%%%%%%%%%%%%%%%%%%%%%%%%%%%%%%%%%%%%%%%%%%%%%%%%%%%%%%%%%%%%%%
The pairwise likelihood is particularly useful for modelling correlated binary outcomes, as discussed in \cite{le1994logistic}. Correlated binary data typically arise in the context of repeated measurements on the same individual. Standard likelihood analysis in these contexts may be difficult because it involves multivariate integrals whose dimension equals the cluster sizes. 
%{Approximate ML estimation via Monte Carlo Expectation Maximisation is possible \citep{chib1998analysis}, but it is computationally intensive and we do not pursue comparisons with it.} {\cambiosnk [Perch\'e questa frase? Non chiedevano MCMC come golden standard?]}

Let us focus on a multivariate probit model with constant cluster sizes. In particular, let $S_i=(S_{i1},\ldots,S_{iq})$ be a latent  normal random variable with mean $\gamma_{i}$ and covariance matrix $\Sigma$, with $\Sigma_{hh}=1+\sigma^{2}$, $\Sigma_{hk}=\sigma^{2}$, $h\neq k$, $h,k=1,\ldots,q$. We assume 
$\gamma_{i}=X_{i}\beta$, where $\beta$ is a vector of unknown regression coefficients and $X_{i}$ is the design matrix for unit $i$, $i=1,\ldots,n$,
Then, the observed data $Y_{ih}$ is equal to 1 if $S_{ih}>0$, and 0 otherwise.

The full likelihood is computationally cumbersome since it entails calculation of multiple integrals of a $q$-variate multivariate normal distribution. On the other hand, the pairwise log-likelihood is
\[
p\ell(\beta,\sigma^2;y) = \sum_{i=1}^{n}\sum_{h=1}^{q-1}\sum_{k=h+1}^{q}\log \Pr(Y_{ih}=y_{ih},Y_{ik}=y_{ik};\beta,\rho), \qquad y_{ih}, y_{ik} \in \{0,1\}\,,
\]
where, for instance,  $\Pr(Y_{ih}=1,Y_{ik}=1;\beta,\rho)=\Phi_{2}(\gamma_{ih},\gamma_{ik};\rho)$ is the standard bivariate normal distribution, with correlation $\rho=\sigma^2/(1+\sigma^2)$ and with $\gamma_{ih}=x_{ih}\beta/\sqrt{1+\sigma^2}$ the $h$th component of $\gamma_{i}$ \citep[see, for instance,][]{cattelan14}.

As an example, we consider data generated with $\beta_{0}=0.5$, $\beta_{1}=1.5$, $\sigma^2=1$,  $n=30$ and $q=10$, where $\beta_0$ is the intercept and $\beta_1$ the coefficient of a covariate generated from the uniform distribution in $(-1,1)$.  For the parameter $\theta = (\beta_0,\beta_1,\log\sigma^2)$ a trivariate normal prior with independent components $N(0, 100)$ is assumed. For ABC we take the counts at each time point $h$, $h=1,\ldots,q$, as a $q$-dimensional summary statistic. Hence, the absolute norm of the difference among the statistics is $\sum_{h=1}^{q}|\sum_{i=1}^{n}(y^{\text{obs}}_{ih}-y_{ih})|$. Other choices of the summary statistic led to less accurate results. For ABC-cs, we consider the rescaled pairwise score, evaluated at $\tilde{\theta}^{\mathrm{obs}}$. The matrices $J(\tilde\theta^{\text{obs}})$ and $H(\tilde\theta^{\text{obs}})$ were computed by simulation with 1000 datasets taken from the model with $\theta = \tilde{\theta}^{\text{obs}}$. 
%The computation of the bivariate normal integrals is performed by calling \texttt{Fortran} routines available from Alan Genz's website (\url{http://www.math.wsu.edu/faculty/genz/homepage}). 
We consider $10^3$ final samples drawn from the ABC and ABC-cs posteriors after fixing $\epsilon$ to the $0.1\%$ quantile of the observed distances. 
The sampling is done via importance sampling, with a multivariate $t$-student importance density, with 5 degrees of freedom, centred at $\tilde{\theta}^{\mathrm{obs}}$ and with scale matrix equal to $5V(\tilde\theta^{\text{obs}})$. We compare the results also with the full posterior approximated by the MCMC method of \cite{chib1998analysis}, and with the pairwise and the calibrated pairwise posteriors approximated by usual MCMC. All MCMC approximations are based on $3\times 10^4$ posterior samples, of which the first 5000 values are discarded.

Figure \ref{fig43-1} shows that the ABC-cs method gives a better approximation than ABC with the chosen summary statistic, when compared to the full posterior computed by MCMC. On the other hand, the non calibrated pairwise posterior is overly concentrated, whereas the calibrated pairwise posteriors of $\beta_0$ and $\log\sigma^2$ seem too dispersed.

\begin{figure}[h!]
\centering
{\includegraphics[width=2in, height=2.15in, angle=-90]{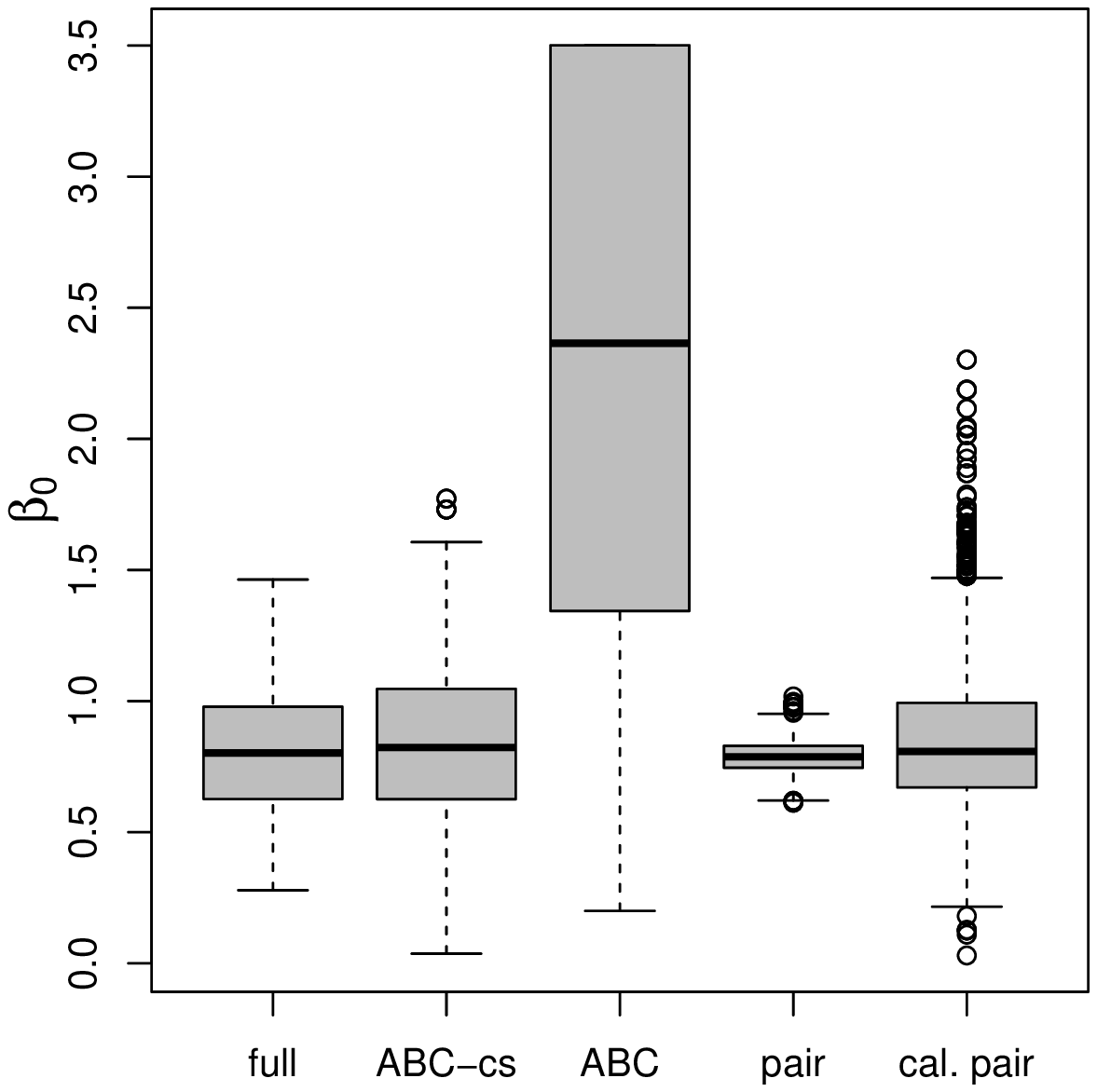}}
{\includegraphics[width=2in, height=2.1in, angle=-90]{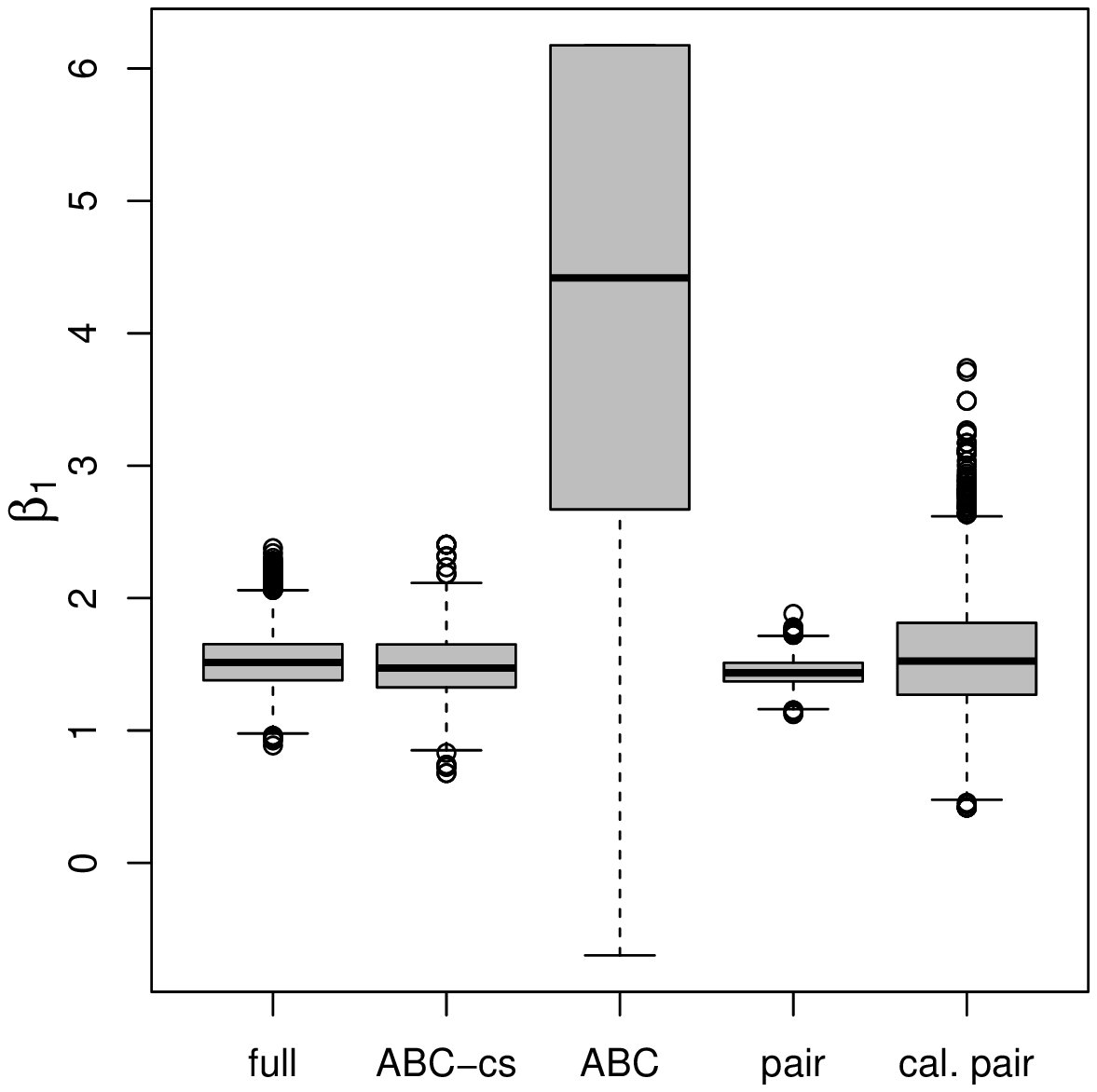}}
{\includegraphics[width=2in, height=2.15in, angle=-90]{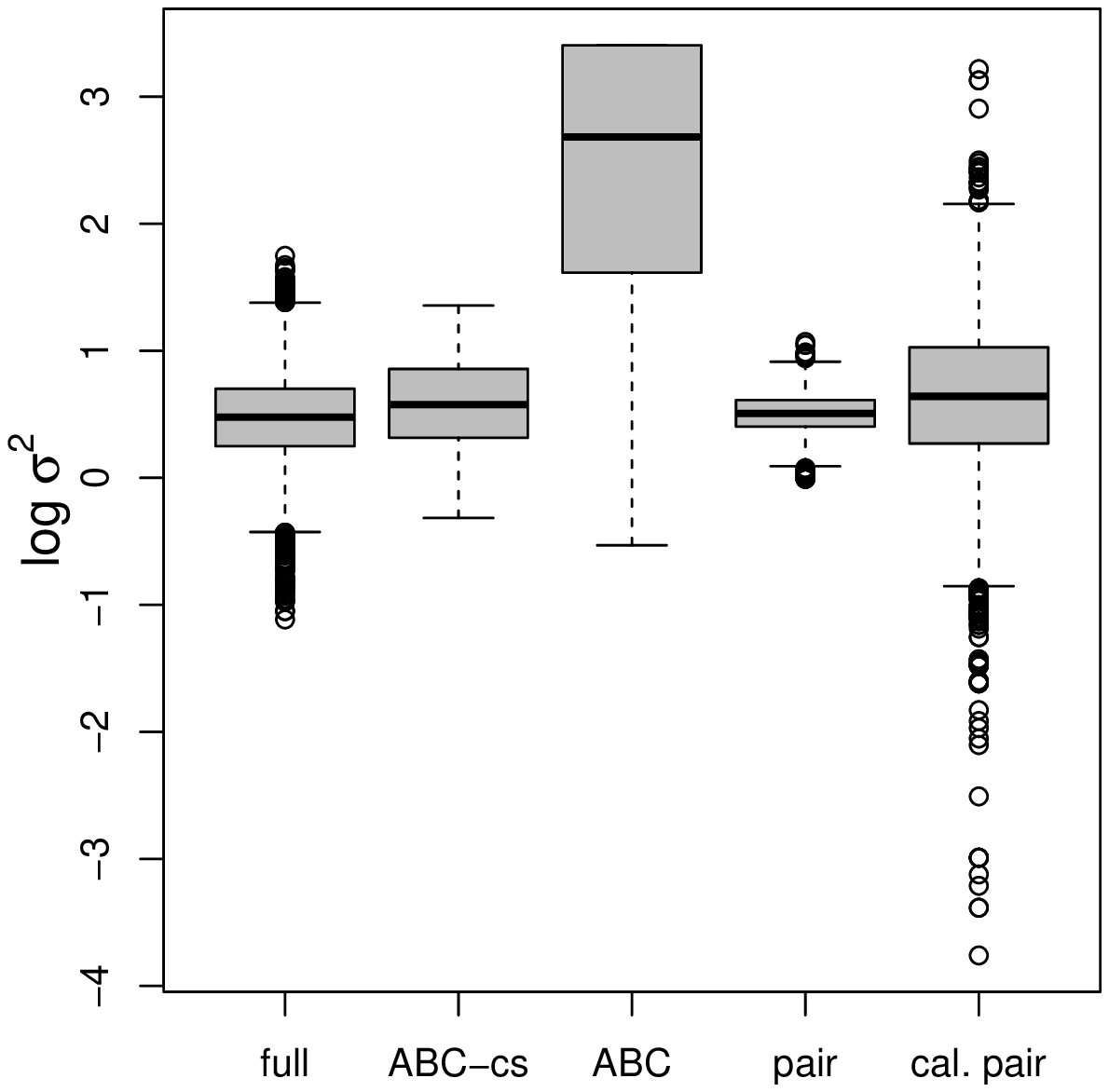}}
\caption{\small Multivariate probit model. ABC-cs posterior compared with the ABC, the pairwise (pair), the calibrated pairwise (cal. pair) and the full posteriors.}
\label{fig43-1}
\end{figure}

A simulation study is conducted over 100 Monte Carlo samples, where the data are simulated as above, with $\beta_0=0.5$, $\beta_1=1.5$, $n=30$, $q=10$ and $\sigma^2=4$. For each simulated dataset, we consider the mean of the ABC, ABC-cs, calibrated pairwise and full posteriors.
Figure~\ref{fig43-2} highlights that the mean of the ABC posterior shows more variability and more bias with respect to the true value (dashed line). On the other hand, the ABC-cs mean is more accurate than the mean of the calibrated pairwise posterior and is in reasonable agreement with the mean of the full posterior. %In this case, also the non calibrated pairwise posterior mean is performing reasonably well. Nevertheless, as shown in Figure \ref{fig43-1}, this posterior is overly precise. 
See the Supplementary Material for additional simulation results with different parameter configurations.\\
  
\begin{figure}[h!]
\centering
{
\includegraphics[width=2in, height=2.15in, angle=-90]{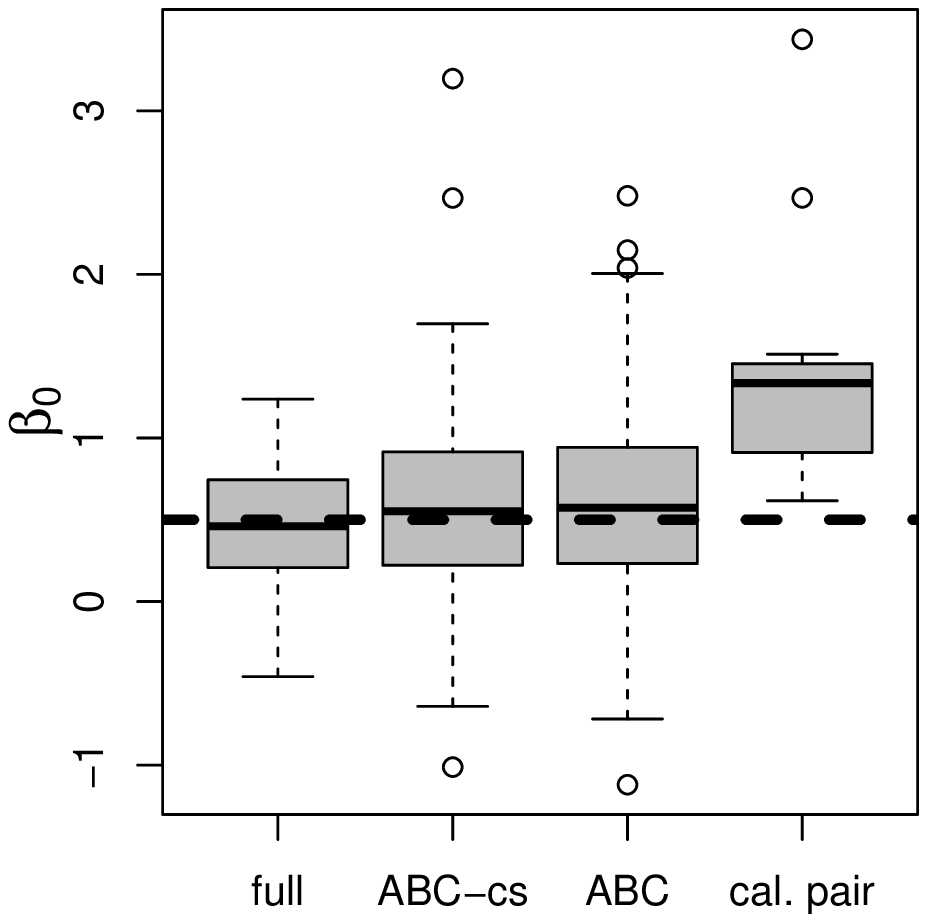}
\includegraphics[width=2in, height=2.1in, angle=-90]{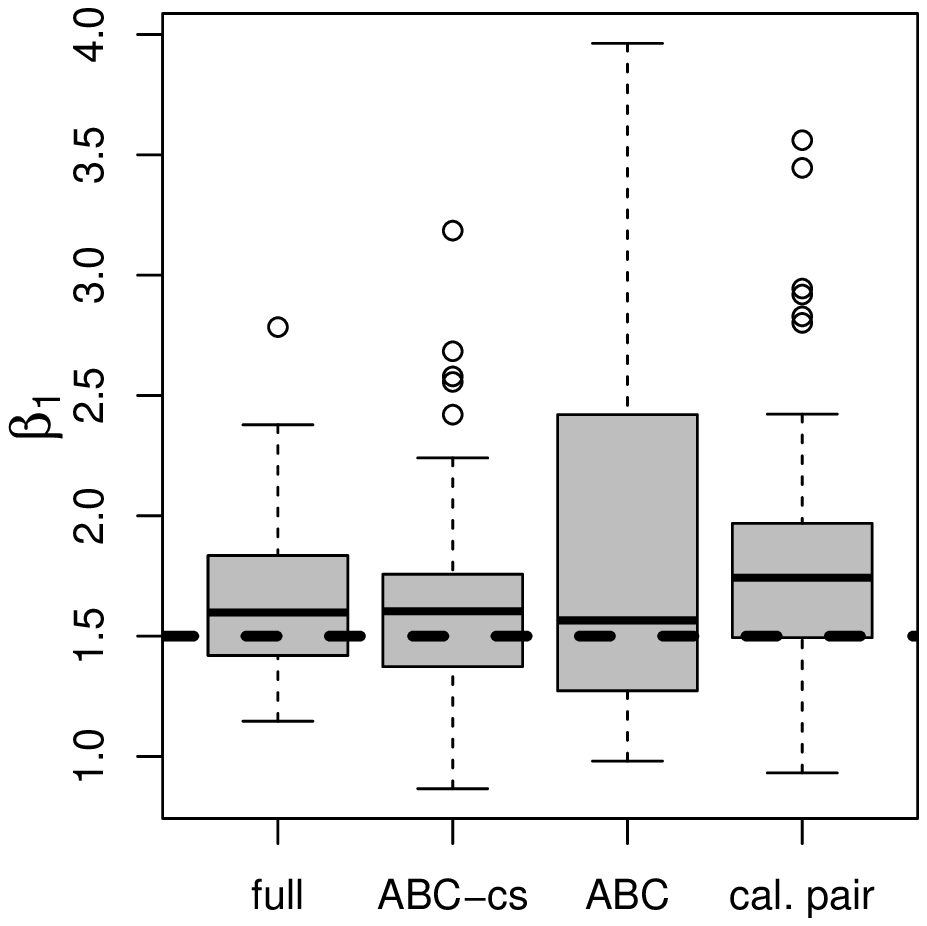}
\includegraphics[width=2in, height=2.15in, angle=-90]{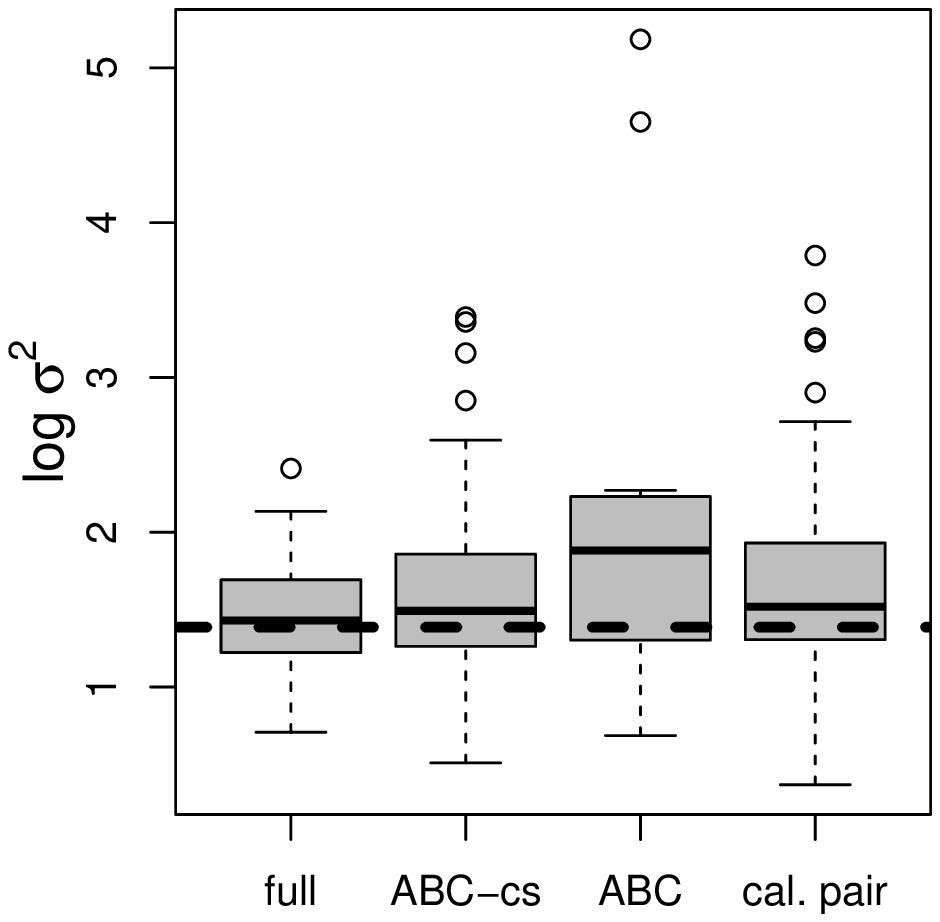}}\\
%\multicolumn{3}{c}{(c) $\sigma^2=4$ ($\log\sigma^2=1.39$)}\\ 
\caption{\small Multivariate probit model. Simulations based on 100 Monte Carlo trials, with $\beta_0=0.5$, $\beta_1=1$ and $\sigma^2=4$ ($\log\sigma^2=1.39$). }\label{fig43-2}
\end{figure}

%%%%%%%%%%%%%%%%%%%%%%%%%%%%%%%%%%%%%%%%%%%%%%%%%%%%%%%%%%%%%%%%
\section{Spatial extremes}
%%%%%%%%%%%%%%%%%%%%%%%%%%%%%%%%%%%%%%%%%%%%%%%%%%%%%%%%%%%%%%%%

Understanding and modelling the behaviour of natural phenomena such as heat waves, heavy rainfall or air pollution can be of interest for climate, social and statistical scientists, to stakeholders such as insurance companies and public health officials. It is therefore important to have useful statistical methods for modelling these extreme occurrences and assessing their possible consequences and impacts. 

As these phenomena materialize in spatio-temporal contexts, a natural approach to their modelling is through the theory of max-stable processes \citep[see][]{schlather2002, kabluchko2009, deHaan1984}, an infinite-dimensional extension to multivariate extreme value theory. Max-stable modelling has the potential advantage of accounting for spatial dependence of extremes in a way that is consistent with the classical extreme-value theory, but is much less well developed than other competitive approaches \citep{davison2012statistical}. Some applications to rainfall data can be found in \cite{buishand2008spatial}, \cite{smith2009extended}, \cite{padoan2010likelihood}, \cite{davison2012statistical}, \cite{ribatet.etal2012}, to temperature data in \cite{davison2012geostatistics}, and to snowfall data in \cite{blanchet2011spatial}.

Here we focus on the popular max-stable process introduced by \cite{smith1990}. Full description of this model can be found in \cite{padoan2010likelihood} and \cite{davison2012statistical}, to which we refer for the details. The bivariate marginal distribution of Smith's model at spatial coordinates $t_k, t_l\in\Real^2$, for $k\neq l ={1,\ldots, q}$, with $q$ being the number of spatial locations, is
\[
\Pr\{Z(t_k)\leq z_k,Z(t_l)\leq z_l\} = \exp\left[-\frac{1}{z_k}\Phi\left(\frac{a(h)}{2}+\frac{1}{a(h)}\log\frac{z_l}{z_k}\right)-\frac{1}{z_l}\Phi\left(\frac{a(h)}{2}+\frac{1}{a(h)}\log\frac{z_k}{z_l}\right)\right],
\]
where $h=(t_l-t_k)$, $a(h) =(h^\T\Sigma^{-1}h)^{1/2}$, $\Sigma$ is the covariance matrix of the process with variances $\sigma_1^2,\sigma_2^2>0$ and covariance $\sigma_{12}$. 
The corresponding density function is obtained by straightforward differentiation \citep[see, e.g.,][]{padoan2010likelihood}.

An expression for the trivariate marginal density of Smith's model is derived by \cite{genton2011likelihood}. However, there is no closed form expression for marginal densities of dimension greater than three and so the full likelihood is intractable. Pairwise likelihood inference is therefore a natural approach in this context, and it was first advocated by \cite{padoan2010likelihood}; see also \cite{blanchet2011spatial}, \cite{sang2014tapered}, \cite{ribatet.etal2012} and \cite{smith2009extended}. Although the triplewise likelihood can be more efficient than the pairwise likelihood \citep{genton2011likelihood}, for processes typically used in applications the efficiency gains are not striking and therefore the pairwise likelihood provides a good compromise between statistical and computational efficiency \citep{huser2013composite}.

 %though \cite{genton2011likelihood} show that substantial efficiency gains are possible if trivariate margins are used.

The extremal dependence of Smith's model, and in general for other types of max-stable processes, is typically studied through the so called extremal coefficient \citep{smith1990}. In practice, due to high-dimensional distributional complexity the extremal coefficient is limited to pairwise components. Specifically, for Smith's model such extremal coefficient is $\delta(h)=2\Phi(a(h)/2)$, and the range of the spatial dependence is thus completely governed by $\Sigma$.

Given the data $y_1, \ldots,y_n$, assumed to be $n$ independent replications of the random vector $Y_i\in\Real^q$, $i=1,\ldots,n$, with marginal unit Fr\'echet distribution, \cite{padoan2010likelihood} fit Smith's model by maximising the associated pairwise likelihood. For the generic pair of sites $k, l$ ($k\neq l$), $k,l=1,\ldots,q$, and observation $i$, the pairwise log-likelihood is
\[
p\ell(\theta; y_{ik},y_{il}) = A + \log(BC + D) + \log E\,,
\]
with
$$A = -\frac{\Phi(w(h))}{z_{ik}} - \frac{\Phi(v(h))}{z_{il}},\,\quad B = \frac{\Phi(w(h))}{z_{ik}^2} + \frac{\phi(w(h))}{a(h)z_{ik}^2}-\frac{\phi(v(h))}{a(h)z_{ik}z_{il}},$$
$$C = \frac{\Phi(v(h))}{z_{il}^2}+\frac{\phi(v(h))}{a(h)z_{il}^2}-\frac{\phi(w(h))}{a(h)z_{ik}z_{il}}\,,\quad D=\frac{v(h)\phi(w(h))}{a(h)^2z_{ik}^2}+\frac{w(h)\phi(v(h))}{a(h)^2z_{ik}z_{il}^2}\,,$$
$$ E = \frac{1}{\lambda_k\lambda_t}\left(1+\xi_k\frac{y_{ik}-\mu_k}{\lambda_k}\right)_+^{1/\xi_k -1}\left(1+\xi_l\frac{y_{il}-\mu_l}{\lambda_l}\right)_+^{1/\xi_l -1}\,,$$
$w(h) = a(h)/2+ \log(z_{il}/z_{ik})/a(h)$, $v(h)=a(h)-w(h)$ and $a_+ = \max\{0,a\}$.  Notice that $E$ is essentially the Jacobian due to the standardisation from the observed data $y_{ik}$ to unit Fr\'echet $z_{ik}$, and $\mu_k$, $\lambda_k>0$ and $\xi_k$ are continuous functions, that represent respectively location, scale and shape at site $k$.

Following \cite{padoan2010likelihood},  $\mu_k$ and $\lambda_k$ are assumed as response surfaces, that is $\mu_k=X^\mu_{k}\beta^\mu$, $\lambda_k=X^\lambda_k\beta^\lambda$, where $X^\mu_k$ and $X^\lambda_k$ are vectors of covariates for location $k$, whereas $\beta^\mu\in\Real^{p^\mu}$ and $\beta^\lambda\in\Real^{p^\lambda}$ are unknown regression parameters ($k=1,\ldots,q$). Moreover, for simplicity we assume equal shape among the $q$ locations, e.g. $\xi_k=\xi$, for all $k$. The parameter of this model is $\theta=(\sigma_1^2, \sigma_{12}, \sigma_2^2, \beta^\mu, \beta^\lambda, \xi)$. Other possible models can be constructed by considering spline functions instead of linear regression functions but, for simplicity, here we focus on the latter.

The fitting of max-stable processes for extremes with ABC methods has been first proposed by \cite{erhardt2012approximate}. In particular, they transform the data to unit Gumbel, where the marginal parameters are estimated separately by fitting the Generalised Extreme Value (GEV) distribution at each location by maximum likelihood estimation, and successively estimate the dependence parameters using ABC. The summary statistic proposed by \cite{erhardt2012approximate} is the least square fit of the residuals among the empirical and theoretical pairwise or triplewise madogram. However, ABC-cs allows to estimate jointly both the marginal GEV and tail dependence parameters. Moreover, the rescaled composite score is not computationally as demanding as the summary statistic of \cite{erhardt2012approximate}, which at each simulated data requires a least square fit and a scalar numerical integration.

%\subsection{Swiss rainfall data}
%\label{sec:spexappli}
We illustrate ABC-cs using summer (June to August) maximum daily rainfall data at $q=79$ rain gauging stations located in the north of the Alps and east of the Jura mountains in Switzerland. The dataset is provided by the national meteorological service (M\'et\'eoSuisse) and comprises $n=49$ yearly observations which were derived from daily precipitation data from 1962 to 2008 (a reduced version of this dataset is used also by \citealp{davison2012statistical} and \citealp{sang2014tapered}). The full dataset can be found in the \texttt{R} package \texttt{SpatialExtremes} \citep{spextremes}. An exploratory data analysis reveals that there is some variation in the precipitation with latitude and longitude, which suggests that the process may be anisotropic; see also \cite{davison2012statistical} for a more in-depth description of this dataset.

We set $X^\mu_k=X^\lambda_k = (1, x_k, y_k)$, where $x_k$ and $y_k$ are respectively the latitude and the longitude at location $k$, $k=1,\ldots,79$. The marginal parameters are $\beta^\mu = (\beta^\mu_0, \beta^\mu_1,\beta^\mu_2$), $\beta^\lambda=(\beta^\lambda_0, \beta^\lambda_1, \beta^\lambda_2)$ and the shape is $\xi$; hence $\theta$ has 10 unknown parameters. The \texttt{SpatialExtremes} package is used in order to compute the maximum pairwise likelihood estimate (MPLE) $\tilde\theta^{\text{obs}}$, $H(\tilde\theta^{\text{obs}})$ and $J(\tilde\theta^{\text{obs}})$. The pairwise score function is approximated by finite difference methods \citep[][Section 8.6]{monahan2011numerical}. The prior for $\theta$ is uniform in the space $(0, 1000)\times(-300, 300)\times(0,1000)\times(-\infty,+\infty)^6\times(0, \infty)$, under the condition that $\Sigma$ is a proper covariance matrix. For ABC-cs we use importance sampling with a multivariate $t$ distribution with 5 degrees of freedom, centred at $\tilde\theta^{\text{obs}}$ and with scale matrix $2.5V(\tilde\theta^{\text{obs}})$ as importance function. We draw $1.1\times10^6$ values from the importance density and fix $\epsilon$ to 0.5\% quantile of the observed distances, so we end up with 5500 values from the ABC-cs posterior. 

The results are compared with the MPLE, with the pairwise posterior (5) and with the non calibrated pairwise posterior, the latter two approximated by $5\times 10^4$ MCMC samples. The marginal posteriors for $\sigma_1^2, \sigma_{12}, \sigma_2^2, \beta^\lambda$ are shown in Figure~\ref{fig:spext}, whereas Table~\ref{tab:spext} reports some numerical summaries for all the parameters. The plots of the other marginal posteriors are reported in the Supplementary Material. 
 
\begin{figure}[!h]
\includegraphics[scale=0.42, angle=-90]{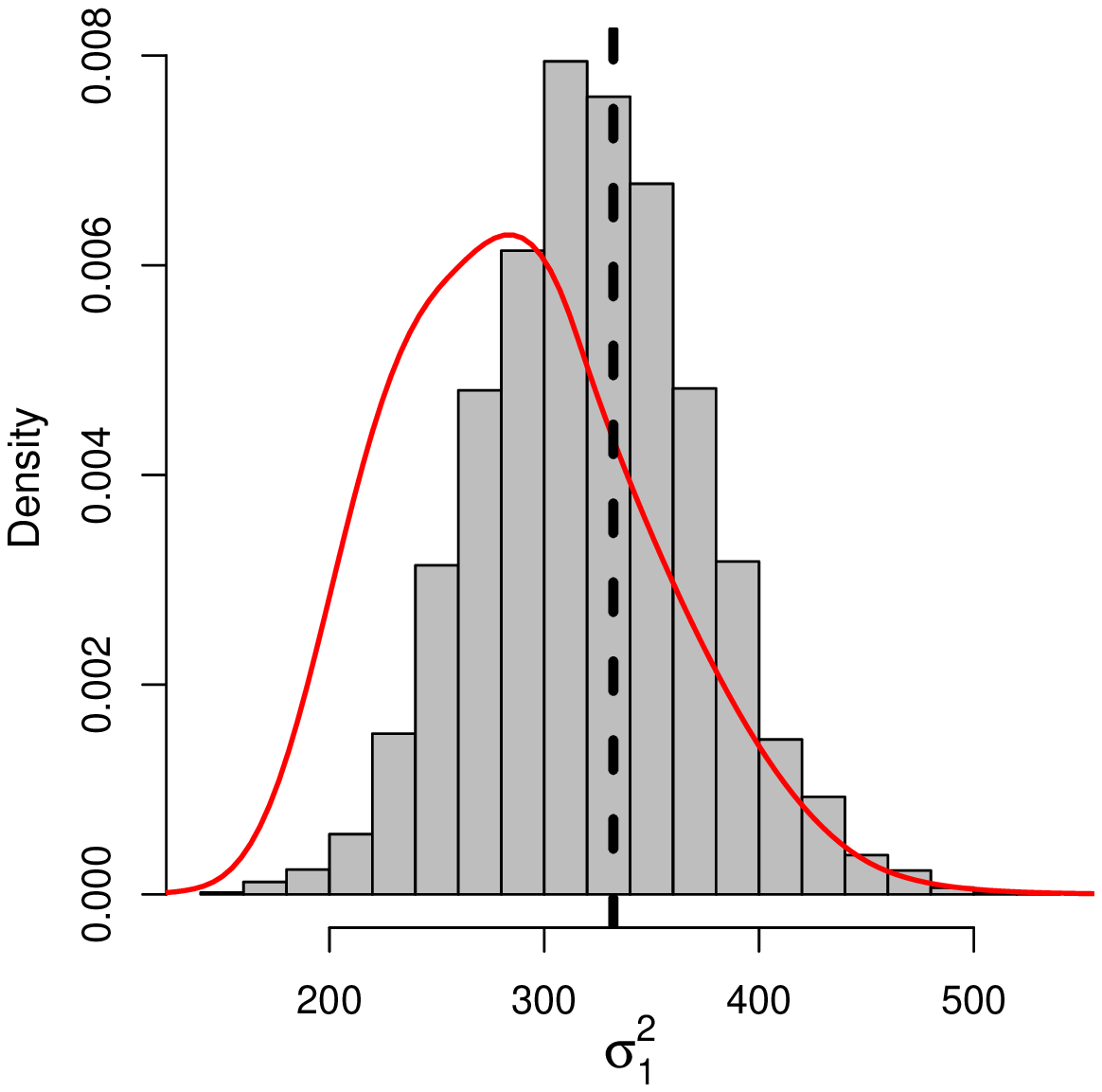}
\includegraphics[scale=0.42, angle=-90]{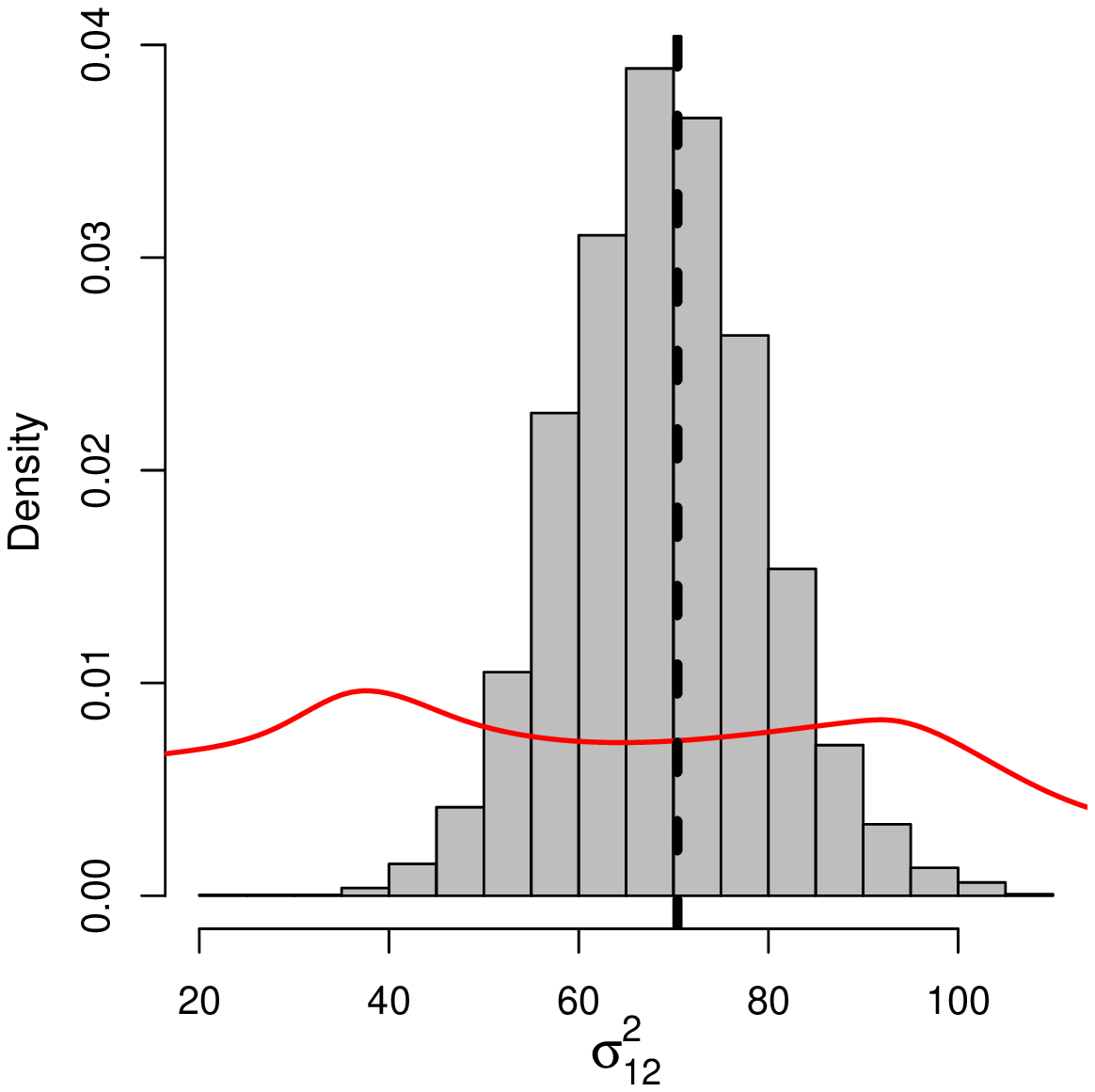}
\includegraphics[scale=0.42, angle=-90]{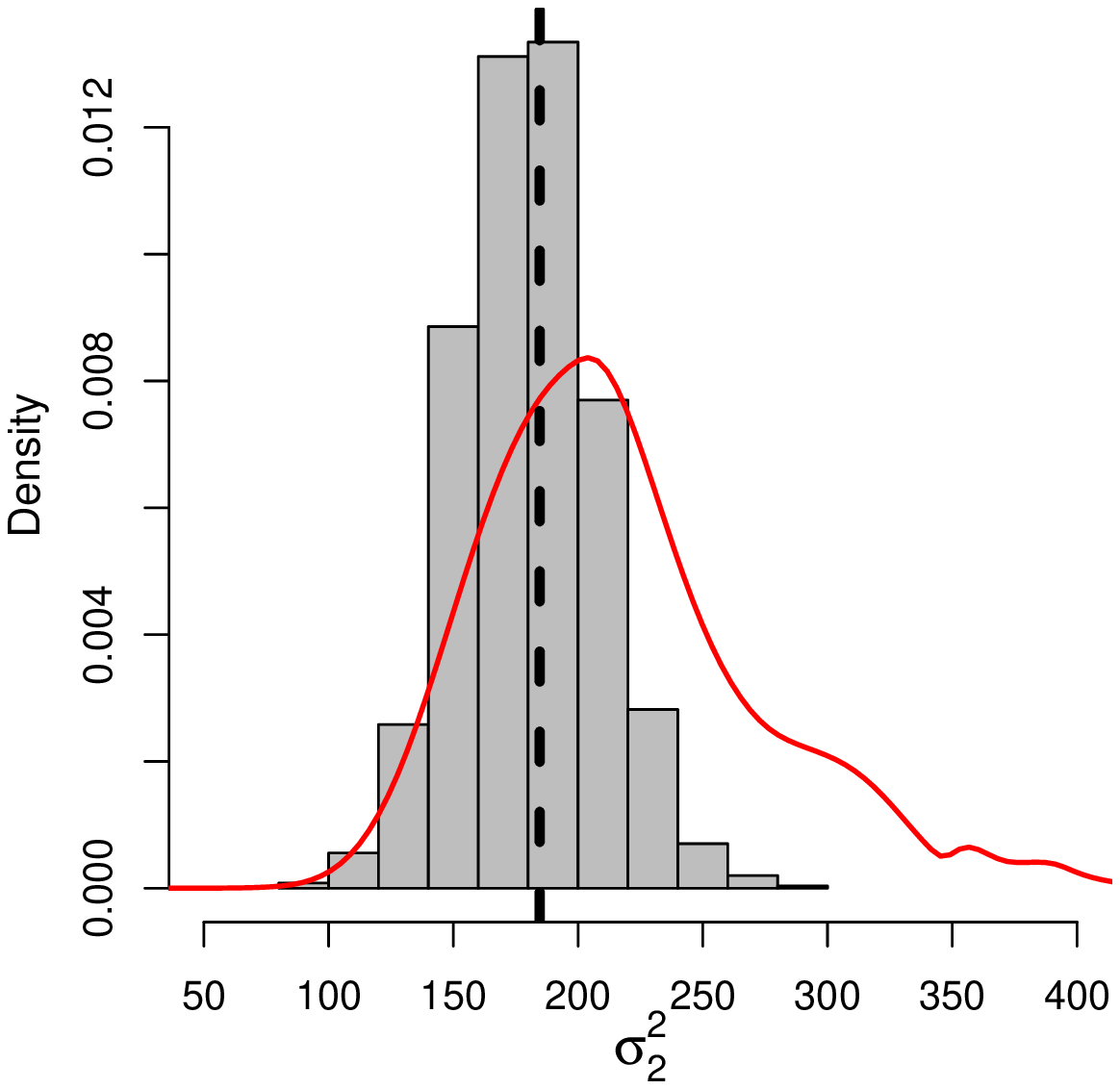}\\
\includegraphics[scale=0.42, angle=-90]{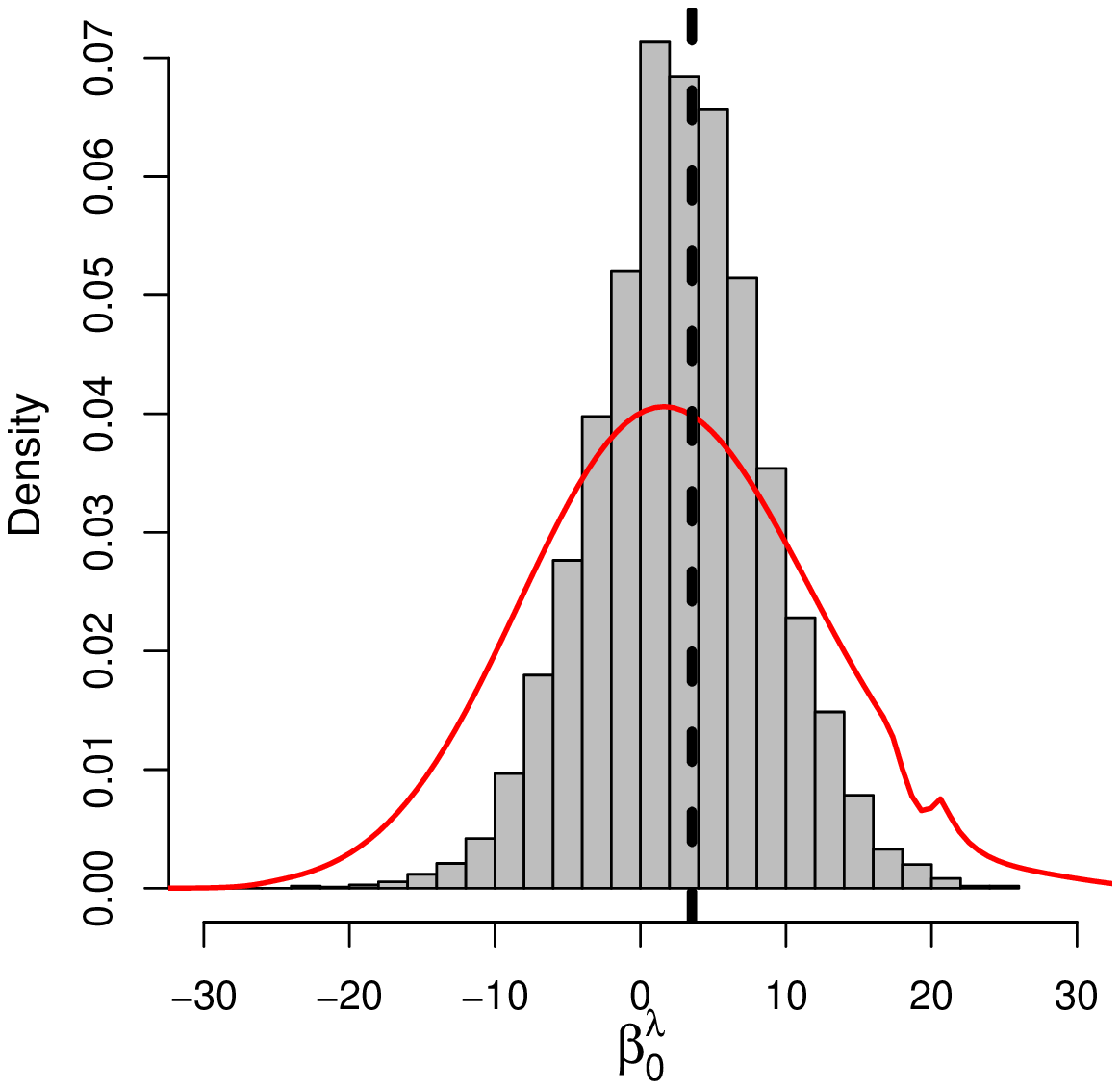}
\includegraphics[scale=0.42, angle=-90]{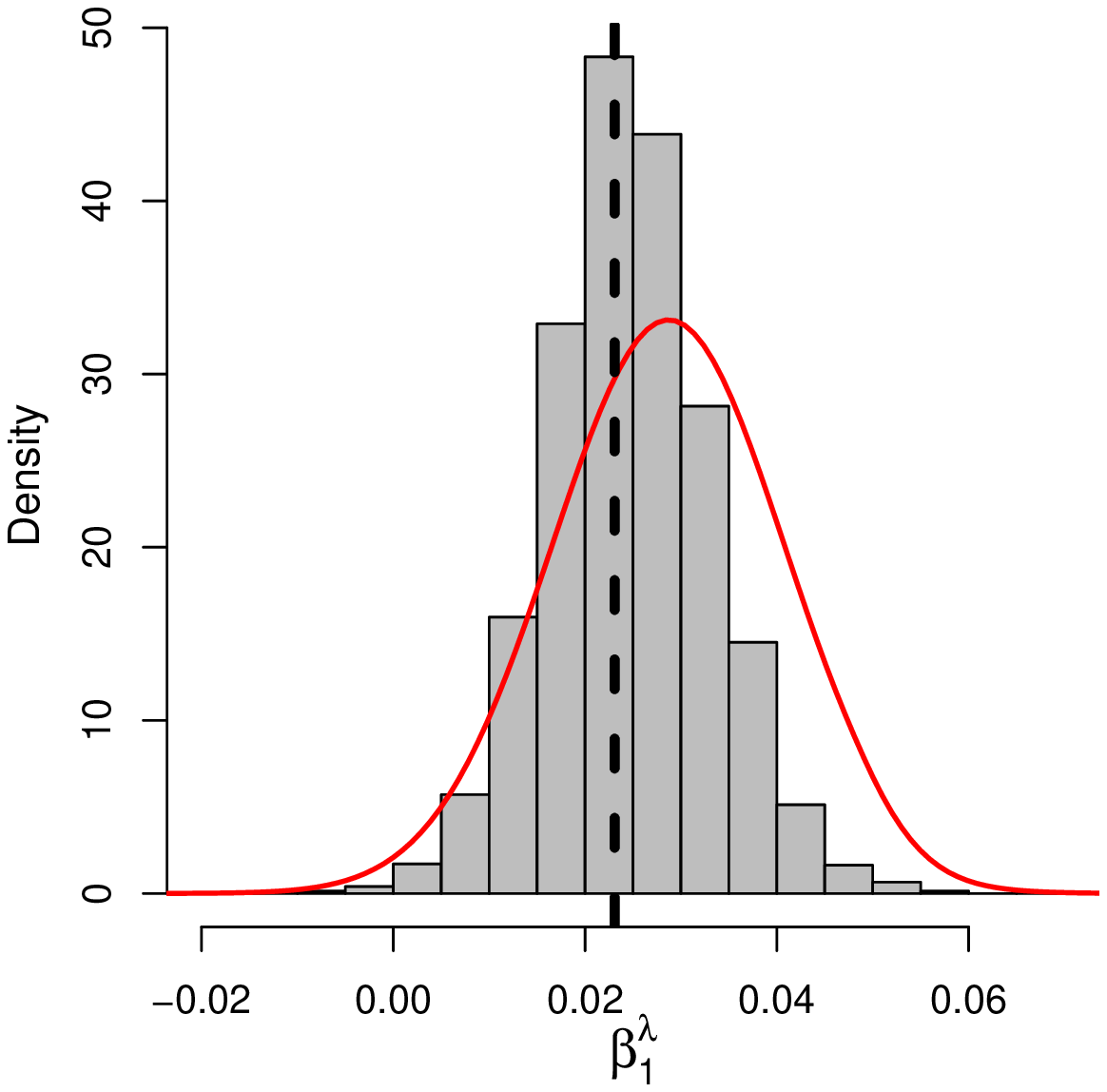}
\includegraphics[scale=0.42, angle=-90]{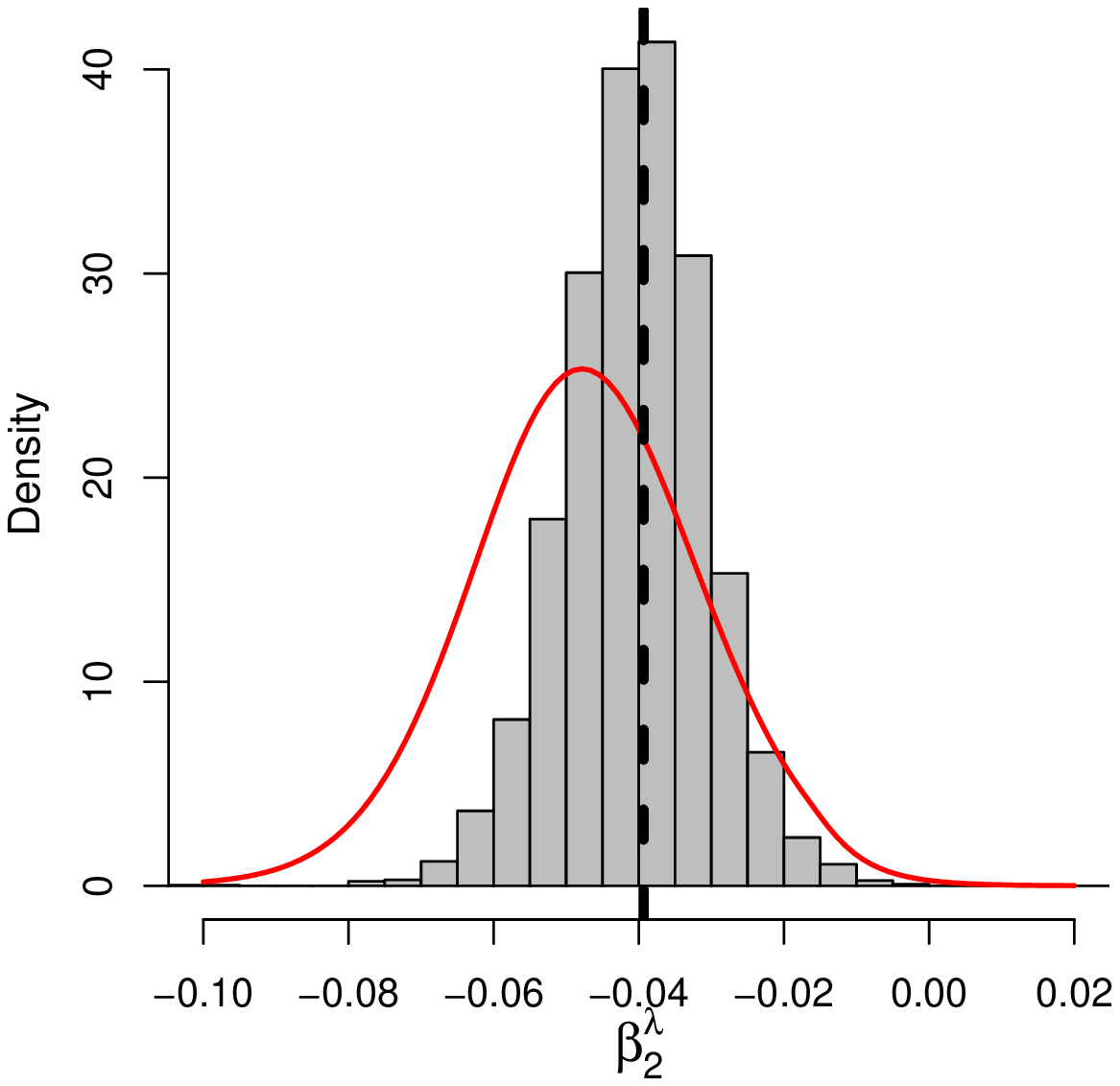}\\
\caption{\small Swiss rainfall data. Marginal ABC-cs (histogram) and calibrated pairwise (kernel densities in red) posteriors compared with MPLE (dashed vertical lines). The first row shows the dependence parameters and the second row shows $\beta^\lambda$.}
\label{fig:spext}
\end{figure}

\begin{table}\centering
\begin{tabular}{crrrr}
\hline
Param. 	 &	 MPLE (SE) & ABC-cs  & Cal. pairwise & Pairwise \\
%    20.65   0.06  -0.16   3.54   0.02  -0.04   0.19
\hline
$\sigma_1^2$ & 332.15 (55.86) & 321.88 (51.03) & 288.71 (61.96) & 325.63 (4.75)\\
$\sigma_{12}^2$ & 70.40 (11.40)& 68.95 (10.41) & 63.39 (40.94) & 69.99 (2.85)\\
$\sigma_2^2$  & 184.63 (30.45)& 180.17 (27.69) & 215.98 (56.51) & 181.49 (2.74)\\
$\beta^\mu_0$ & 20.65 (8.67) & 21.09 (9.09) & 22.62 (12.22) & 22.79 (0.563) \\
$\beta^\mu_1$ & 0.06 (0.01)	 & 0.06 (0.01) & 0.06 (0.02) &  0.06 ( 0.001)\\
$\beta^\mu_2$ & -0.16 (0.02) & -0.15 (0.02) & -0.15 (0.02) & -0.15 ( 0.001)\\
$\beta^\lambda_0$ & 3.54 (5.56) & 2.65 (5.97) & 2.01 (9.62) & 2.35 (0.43)\\
$\beta^\lambda_1$ & 0.02 (0.01) & 0.02 (0.01) & 0.03 (0.01) & 0.03 (0.001)\\
$\beta^\lambda_2$ & -0.04 (0.01) & -0.04 (0.01)& -0.05 (0.02) & -0.05 (0.001) \\
$\xi$             &  0.19 (0.03) & 0.18 (0.03) & 0.18 (0.03) & 0.18 (0.001)\\
\hline
\end{tabular}
\caption{\small Swiss rainfall data. Means (and standard deviations) of ABC-cs, pairwise and calibrated pairwise posteriors, compared with the MPLE and its asymptotic standard error (SE).}\label{tab:spext}
\end{table}

Table~\ref{tab:spext} confirms that Bayesian inference based on the non calibrated pairwise likelihood can be overly too precise. The calibrated pairwise and the ABC-cs posteriors for the marginal parameters appear to be quite different, especially for the dependence parameters (first row of Figure~\ref{fig:spext}). For instance, while the ABC-cs 0.95 credible interval for $\sigma_{12}$ does not include zero, the contrary holds for MCMC calibrated pairwise credible interval; the latter suggests that the process may be isotropic. The ABC-cs posterior is in good agreement with the MPLE and the corresponding standard errors, whereas the calibrated pairwise posterior appears to be too dispersed.
\begin{figure}[h!]
\includegraphics[angle=-90, scale=0.60]{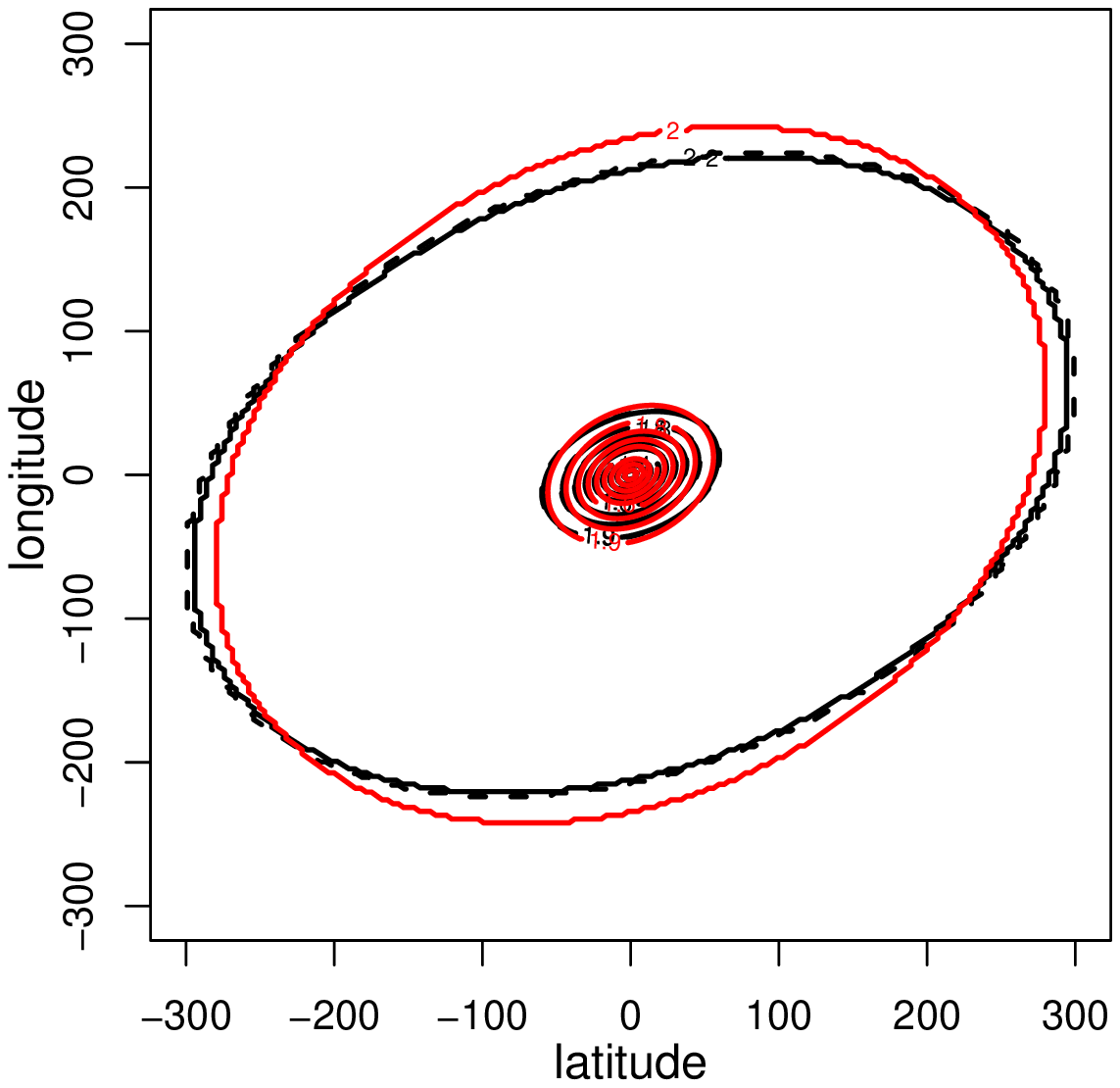}
\includegraphics[angle=-90, scale=0.60]{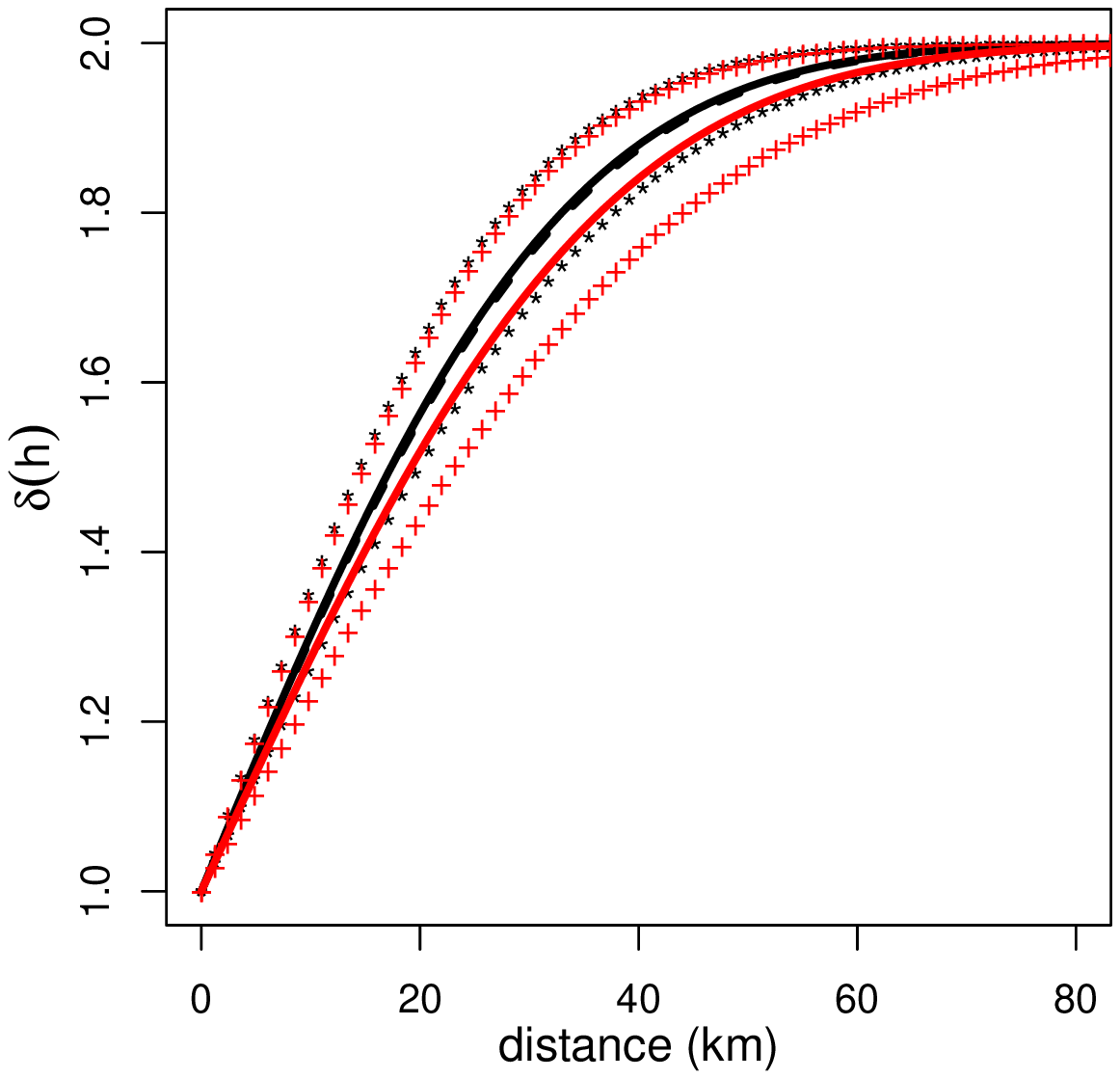}
\caption{\small Left panel: pairwise extremal coefficient of the Smith model with the Swiss rainfall data computed using the ABC-cs posterior mean (continued line), the calibrated pairwise posterior mean (red line) and the MPLE (dashed). Right panel: pairwise extremal coefficient plotted as function of the Euclidean distance among the locations, with 0.95 pointwise credible bands from ABC-cs (*) and calibrated pairwise posterior (+).}
\label{fig:extremal}
\end{figure}

We further compare the MPLE with the ABC-cs and the calibrated pairwise posteriors by plotting the estimated pairwise extremal coefficient calculated at the corresponding means. These comparisons are shown on the left panel of Figure~\ref{fig:extremal}. Again we notice that ABC-cs is very similar to the MPLE, whereas the extremal coefficient based on the calibrated pairwise posterior appears substantially different. Moreover, this plot confirms that the calibrated pairwise posterior shows more isotropy than MPLE or ABC-cs. The right panel of Figure~\ref{fig:extremal} shows the extremal coefficient as a function of the Euclidean distance among the locations, with 0.95 credible bands computed pointwise at each distance of $h$. This plot highlights that the extremal coefficients computed from the calibrated pairwise posterior show more variability than those obtained from ABC-cs. The Wald-type confidence bands of MPLE essentially overlap with those of ABC-cs and therefore are not reported.

%%%%%%%%%%%%%%%%%%%%%%%%%%%%%%%%%%%%%%%%%%%%%%%%%%%%%%%%%%%%%%%%
\section{Discussion}
%%%%%%%%%%%%%%%%%%%%%%%%%%%%%%%%%%%%%%%%%%%%%%%%%%%%%%%%%%%%%%%%

A new procedure for constructing summary statistics for ABC is proposed, which is based on a rescaled  composite score function. An advantage of the proposed method is that, by construction, the summary statistics automatically incorporate relevant features of the complex model, and its dimension is the same as the number of parameters. Moreover, no post processing tasks are required, nor pilot runs or {\em ad hoc} summaries of the data. With a little additional computational effort, the Godambe information can be obtained as a by-product of our method. Such information matrix can be used as a scaling matrix in simulation schemes.

{Although being computationally more expensive than Bayesian composite posteriors, ABC-cs does not require calibration. Moreover, as seen from the examples and from the application, Bayesian inference with composite likelihoods, both calibrated and non calibrated, can be quite inaccurate.}

The ABC-cs procedure depends of course on the availability of a reasonable composite likelihood for the problem under investigation, which may not be easy to obtain, or even define, in some problems. 
However, there is a rich and growing literature on composite likelihoods \citep{varin2011overview,larribe2011composite}, which we believe may be fruitfully used to identify the class of problems in which composite likelihoods may be used in ABC and also to guide the choice of the more appropriate composite likelihood.
In general, a sensible composite likelihood has to be a good approximation for the full likelihood, or at least it has to appropriately describe the main features of interest of the model, by keeping a reasonable computational complexity. Even with this in mind,  there could still be possible competing composite likelihoods for the same model. Recent contributions in  frequentist inference consider the idea of combining different composite likelihoods in order to improve efficiency \citep{cox2004note,kennepagui2014}. This idea could be used also in ABC-cs, and further extended to the combination of composite scores with other summary statistics.

% at hand {\cambios REF 1 punt 13: charie. Anche per REF 2 punt 23}. {\cambiosE Sometimes, more than one composite likelihoods may be available, and akin to ABC, one is confronted with the issue of choosing among them. Nonetheless, in the composite likelihood framework this issue is less severe as several solutions are available. A first solution is to approach the problem \'a la \cite{heggland2004}, where the best estimating function is such that has little variance for fixed $\theta$ and high derivative for a given function value. An other option is to combine composite likelihoods together by trying to maximise some suitable functions of the expected or observed Fisher information provided by such combinations \cite{cox2004note}. Yet, another option could be the approach of, where the summary statistics are the various composite score functions. It may be interesting to compare such approaches in realistic settings, such as that of the application in Section~5, but this is beyond the aim of the present work.} There is a rich and growing literature on composite likelihood \citep{varin2011overview,larribe2011composite}, which we believe may be fruitful for ABC applications.  

Finally, we note that we used the composite likelihood as a natural basis to construct a suitable unbiased estimating function in complex models. However, the proposed ABC algorithm works with any unbiased estimating function, such as for instance those used in the robust literature \citep[see, e.g.,][]{huber2009robust}.

%{\cambios Estendere discussione. Vedi REF 2 punto 24. E combining composite likelihoods}
{\section*{Supplementary Material}
The online Supplementary Material includes additional simulations and plots for the examples of Section 4 and the application in Section 5, and another example with a moving average process of order 2. The \texttt{R} code for the examples and the application is also included.
}

{\section*{Acknowledgements}

We are grateful to two anonymous Reviewers for their thoughtful comments which lead to a substantially improved version of our initial draft. We thank Manuela Cattelan for sharing her code on the multivariate probit example. This work was supported by a grant from the University of Padua (Progetti di Ricerca di Ateneo 2013) and by the Cariparo Foundation Excellence-grant 2011/2012.}
 
%%%%%%%%%%%%%%%%%%%%%%%%%%%%%%%%%%%%%%%%%%%%%%%%%%%%%%%%%%%%%%%%
\bibliographystyle{biometrika}
\bibliography{Bibliography}
%%%%%%%%%%%%%%%%%%%%%%%%%%%%%%%%%%%%%%%%%%%%%%%%%%%%%%%%%%%%%%%%
\end{document}